\documentclass[twocolumn,10pt]{asme2ej}

\usepackage{amsmath,amsthm,graphicx,amsfonts,amssymb,epsfig,mathrsfs,mathtools,esvect,gensymb, derivative,multirow,tabularx}
\usepackage[caption=false,font=footnotesize]{subfig}
\usepackage{comment}
\usepackage[linesnumbered, ruled, vlined]{algorithm2e}
\usepackage{placeins}
\usepackage{xcolor}
\usepackage[bookmarks=true]{hyperref}
\usepackage{xfrac}
\usepackage{xurl}
\usepackage{wrapfig}
\usepackage{arydshln}
\usepackage{blkarray}
\usepackage{lipsum}
\usepackage{dblfloatfix}
\usepackage{tikz}
\usepackage{color,soul}
\usepackage[normalem]{ulem}
\usepackage{cancel}
\date{}

\newcommand{\xx}{\mathsf{x}}
\newcommand{\y}{\mathsf{y}}
\newcommand{\s}{\mathcal{S}}
\newcommand{\w}{\mathcal{W}}
\newcommand{\D}{\mathcal{D}}
\newcommand{\norm}[1]{\left\|{#1}\right\|}

{}
{}
\newtheorem{remark}{Remark}{}
\newtheorem{example}{Example}{}
{}
{}
\newtheorem{proposition}{Proposition}{}
\newtheorem{theorem}{Theorem}{}

\DeclareMathOperator*{\argmin}{arg\,min}
\usepackage{graphicx}      
\usepackage{cancel}
\usepackage{booktabs} 
\usepackage{hyperref}   
\hypersetup{
	colorlinks=true,
	linkcolor=blue,
	citecolor=blue,
	urlcolor=blue,
}
\usepackage[square,numbers]{natbib}

\title{On the Convergence of Density-Based Predictive Control for Multi-Agent Non-Uniform Area Coverage}

\author{Sungjun Seo\\
    \affiliation{
        PhD Candidate\\
	Department of Mechanical Engineering\\
	New Mexico Institute of Mining and Technology\\
	Socorro, New Mexico 87801\\
    Email: sungjun.seo@student.nmt.edu
    }	
}

\author{Kooktae Lee\thanks{Corresponding author. This paper has been accepted for publication in the ASME Journal of Dynamic Systems, Measurement, and Control (JDSMC).  © 2025 by the American Society of Mechanical Engineers (ASME).} \\
    \affiliation{Associate Professor}\\
	Department of Mechanical Engineering\\
	New Mexico Institute of Mining and Technology\\
	Socorro, New Mexico 87801\\
    Email: kooktae.lee@nmt.edu
}

\begin{document}

\maketitle    

\begin{abstract}
This paper investigates the convergence conditions of Density-based Predictive Control (DPC) for non-uniform area coverage. In large-scale real-world scenarios, such as search and rescue or environmental monitoring missions, efficient non-uniform multi-agent area coverage is essential, as uniform coverage fails to account for varying regional priorities and operational constraints. To address this, we propose a novel multi-agent density-based predictive control strategy, DPC, grounded in optimal transport (OT) theory. Given a pre-constructed reference distribution representing priority regions, DPC ensures that agents allocate their coverage efforts by spending more time in high-priority or densely sampled areas, achieving effective non-uniform coverage. We analyze the convergence conditions of DPC by formulating the contraction mapping problem in terms of the Wasserstein distance. Additionally, we derive the analytic optimal control law for the unconstrained case and propose a numerical optimization method for determining the optimal control law under input constraints. Comprehensive simulations were conducted on both first-order dynamic systems and a linearized quadrotor model under constrained and unconstrained conditions. The results demonstrate that when the proposed conditions are satisfied, the Wasserstein distance locally converges, and the agent trajectories closely match the non-uniform reference distribution. Furthermore, the comparison with the existing coverage method demonstrated the superiority of the DPC method in non-uniform area coverage.
\end{abstract}

\begin{nomenclature}
\entry{$\mathbb{R}$}{Set of real numbers}
\entry{$\mathbb{N}$}{Set of natural numbers}
\entry{$\mathbb{R}_{\geq 0}$}{Set of nonnegative real numbers}
\entry{$\mathbb{N}_{n:m}$}{Set of natural numbers in the interval $[n, m]$}
\entry{$\mathbb{N}_0$}{Natural numbers including zero, i.e., $\mathbb{N} \cup \{0\}$}
\entry{$\mathbb{R}^n$}{$n$-dimensional real vector space}
\entry{$\mathbb{R}^{n \times m}$}{$n \times m$ real matrix space}
\entry{$A^\top$}{Transpose of vector or matrix $A$}
\entry{$\norm{x}^2_B$}{Quadratic form $x^\top B x$, where $x$ is a vector and $B$ is a square matrix}
\entry{$C^+$}{Moore-Penrose pseudoinverse of matrix $C$}
\entry{$I_n$}{$n \times n$ identity matrix}
\end{nomenclature}

\section{Introduction}

Non-uniform area coverage is critical in various real-world applications where certain regions must be prioritized due to operational constraints or varying levels of importance. For example, in search and rescue (SAR) missions, areas with a higher likelihood of containing survivors or hazardous conditions must be given priority to maximize both efficiency and effectiveness \cite{murphy2014disaster}.
Similarly, in environmental monitoring, regions that are ecologically sensitive or have high pollution levels require increased coverage to ensure accurate data collection and timely intervention \cite{govea2024integration, marjovi2009multi}. These examples highlight the need for strategies that selectively allocate resources based on the varying importance of different regions.

To effectively implement such strategies, multi-agent systems have emerged as a promising solution. By leveraging decentralized coordination and parallelism, multi-agent systems are well-suited for efficiently carrying out area coverage tasks. These systems can significantly improve coverage performance, particularly in large or complex environments, where traditional uniform coverage strategies may fall short. However, non-uniform area coverage presents additional challenges, requiring agents to achieve the desired coverage distribution by prioritizing certain regions while efficiently allocating their limited resources.

A variety of approaches have been proposed to address area coverage problems using multi-agent systems. Traditional strategies primarily focus on uniform coverage, where agents are evenly distributed across the area to be covered. Grid-based methods, such as those outlined in \cite{arkin1993lawnmower} and \cite{xu2014efficient}, ensure that every region is visited at least once, optimizing coverage efficiency in structured environments. However, these methods are less effective in adapting to spatially varying coverage demands, particularly in non-uniform scenarios.

To overcome these limitations, recent research has increasingly focused on non-uniform area coverage strategies. Spectral Multiscale Coverage (SMC) \cite{GM-IM:11} applies ergodic principles and approximates the target density using Fourier basis functions. While SMC effectively guides agents toward high-priority regions, it is centrally developed and requires truncation of the infinite Fourier series, leading to practical approximation errors and scalability limitations \cite{silverman2013optimal, lee2018receding}. Voronoi-based coverage control \cite{schwager2008consensus, schwager2015robust} redistributes agents’ final positions according to a given density, but it does not consider the time-averaged behavior of agents, limiting its applicability to persistent coverage tasks. Optimal Transport (OT)-based methods \cite{kabir2020receding, kabir2021wildlife, kabir2021efficient, lee2022density} align agent trajectories with a reference distribution; however, most prior work ignores agent dynamics and provides no formal guarantees of optimality. 
Although subsequent studies \cite{seo2025tcst, seo2025tsmcs} have extended OT-based approaches to account for agent dynamics through an optimal control framework, they do not address stability analysis of the controlled trajectories.
More recently, reinforcement learning (RL)-based strategies \cite{meng2021deep} and decentralized optimization techniques \cite{yang2020distributed} have been explored to enable adaptive coverage behaviors under dynamic constraints, yet RL methods often require extensive training and may lack interpretability, while decentralized optimization can face convergence and stability challenges.

It is also important to distinguish Traveling Salesman Problem (TSP)-based approaches \cite{kalmar2017multiagent, guo2024imtsp}, which focus on planning efficient tours through discrete waypoints. These methods are effective for structured path planning problems involving finite sets of targets. However, they do not address spatially continuous or density-aware coverage tasks, and therefore differ in both objective and applicability from uniform or non-uniform area coverage methods.

Although the aforementioned approaches have advanced non-uniform coverage, many still do not fully address critical aspects of optimality in realistic multi-agent scenarios. Practical limitations, such as the number of agents, operational time, and energy constraints, are often overlooked. In addition, theoretical guarantees regarding convergence and stability are frequently absent, raising concerns about robustness in real-world deployments.

To address these shortcomings, we propose a novel control scheme, Density-based Predictive Control (DPC), for non-uniform area coverage in multi-agent systems. Grounded in OT theory, DPC incorporates agent dynamics and guides trajectories to align with a reference distribution representing priority regions, concentrating efforts on high-priority areas while addressing the practical and theoretical limitations of existing methods. The main contributions of this paper are summarized as follows:

\begin{itemize}
\item[(i)] We propose a control scheme that enables multiple agents to achieve collaborative, non-uniform area coverage while accounting for practical constraints such as battery limitations, agent availability, and communication range.

\item[(ii)] The proposed DPC method performs spatially continuous coverage over the task space, unlike existing OT-based methods that rely on discrete waypoint selection from sample points for path planning.

\item[(iii)] In contrast to prior OT-based approaches, which do not derive an explicit control law and ignore agent dynamics during trajectory generation, our method formulates and solves an optimal control problem that incorporates both unconstrained and constrained agent dynamics to generate dynamically feasible control inputs.

\item[(iv)] We establish a convergence condition for the derived control input using the local Wasserstein distance, ensuring provable and reliable area coverage performance.

\item[(v)] The effectiveness of the proposed method is validated through simulations on both a first-order system and a linearized quadrotor model, where the resulting agent trajectories closely match the reference distribution while satisfying the derived convergence conditions.
\end{itemize}

\vspace{-0.1in}
\section{Problem description}
\subsection{Area Coverage Problem}

In many practical applications, achieving spatial coverage of a designated area is crucial for desired performance. In this study, the main goal is to deploy a set of agents so that the area of interest is covered as efficiently and effectively as possible.

\textbf{Illustrative Application: Victim Detection in Search and Rescue Missions}

In the aftermath of large-scale natural disasters, rapid and efficient victim detection is critical to minimizing casualties and ensuring timely medical assistance. For instance, following the 2011 Tōhoku earthquake and tsunami in Japan, widespread destruction left thousands trapped under debris, with emergency responders struggling to locate survivors across vast and hazardous terrain. A well-coordinated victim detection mission is essential, as rapid identification of survivors significantly improves the chances of a successful rescue. 

A conventional approach to such missions involves uniformly scanning the affected region. However, this method is highly inefficient due to the scale of the disaster area and the limited availability of search resources, often resulting in a superficial search that does not prioritize high-risk zones. This approach fails to account for the spatial distribution of likely survivor locations, leading to inefficient search efforts and wasted time.

A more effective victim detection strategy incorporates probabilistic information to guide search operations. Factors such as building collapse density, historical disaster data, real-time aerial imagery, and distress signals from mobile devices can be integrated to construct a probabilistic map of likely victim locations. As illustrated in Fig.~\ref{fig: concept}(a), a set of discretized reference points can be pre-constructed to represent the likelihood of victim presence using these diverse information sources.
\begin{figure}[h]
    \centering
    \subfloat[Reference point distribution representing victim likelihood]{
    \includegraphics[width=0.38\linewidth]{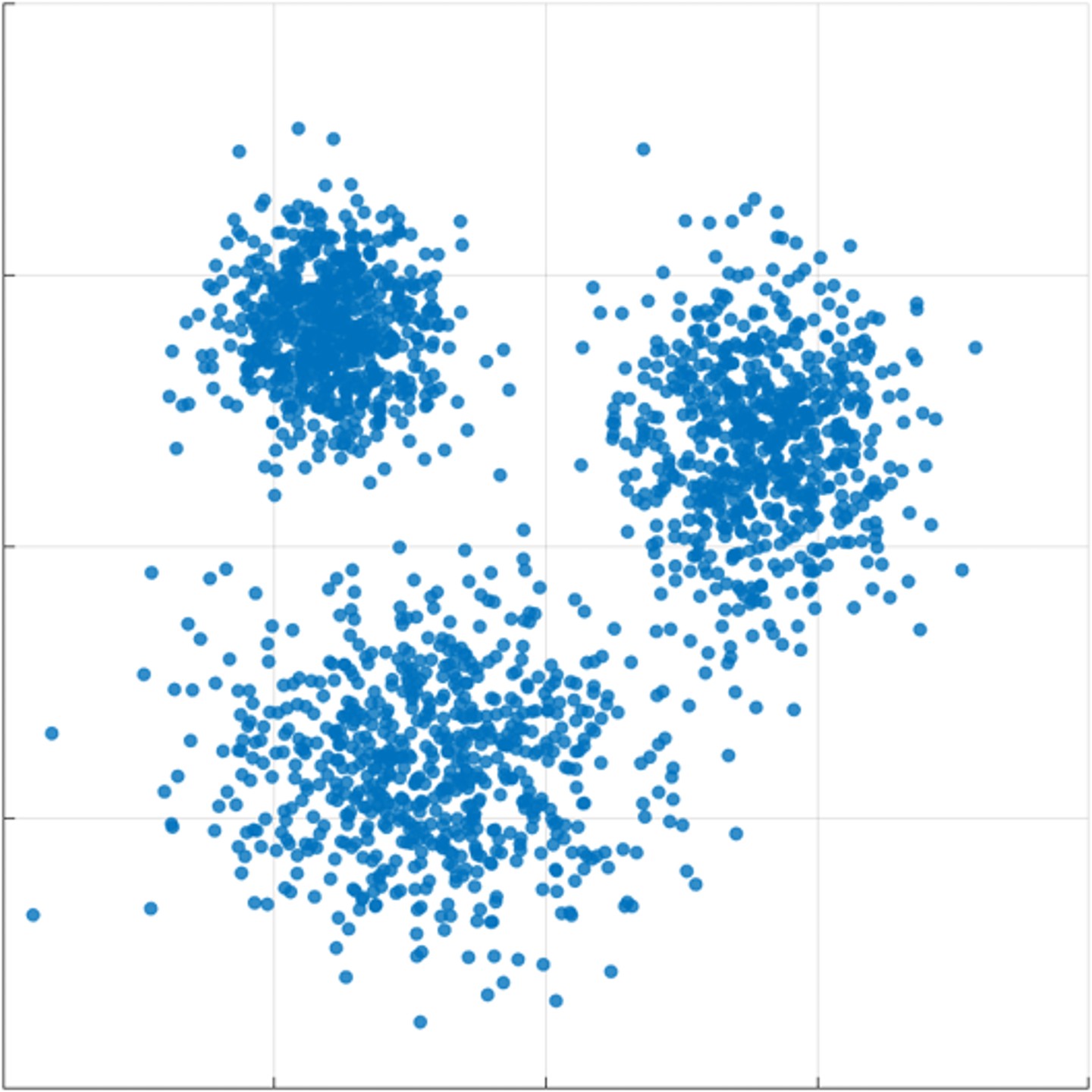}}\quad
    \subfloat[Efficient search coverage by the multi-agent system]{
    \includegraphics[width=0.38\linewidth]{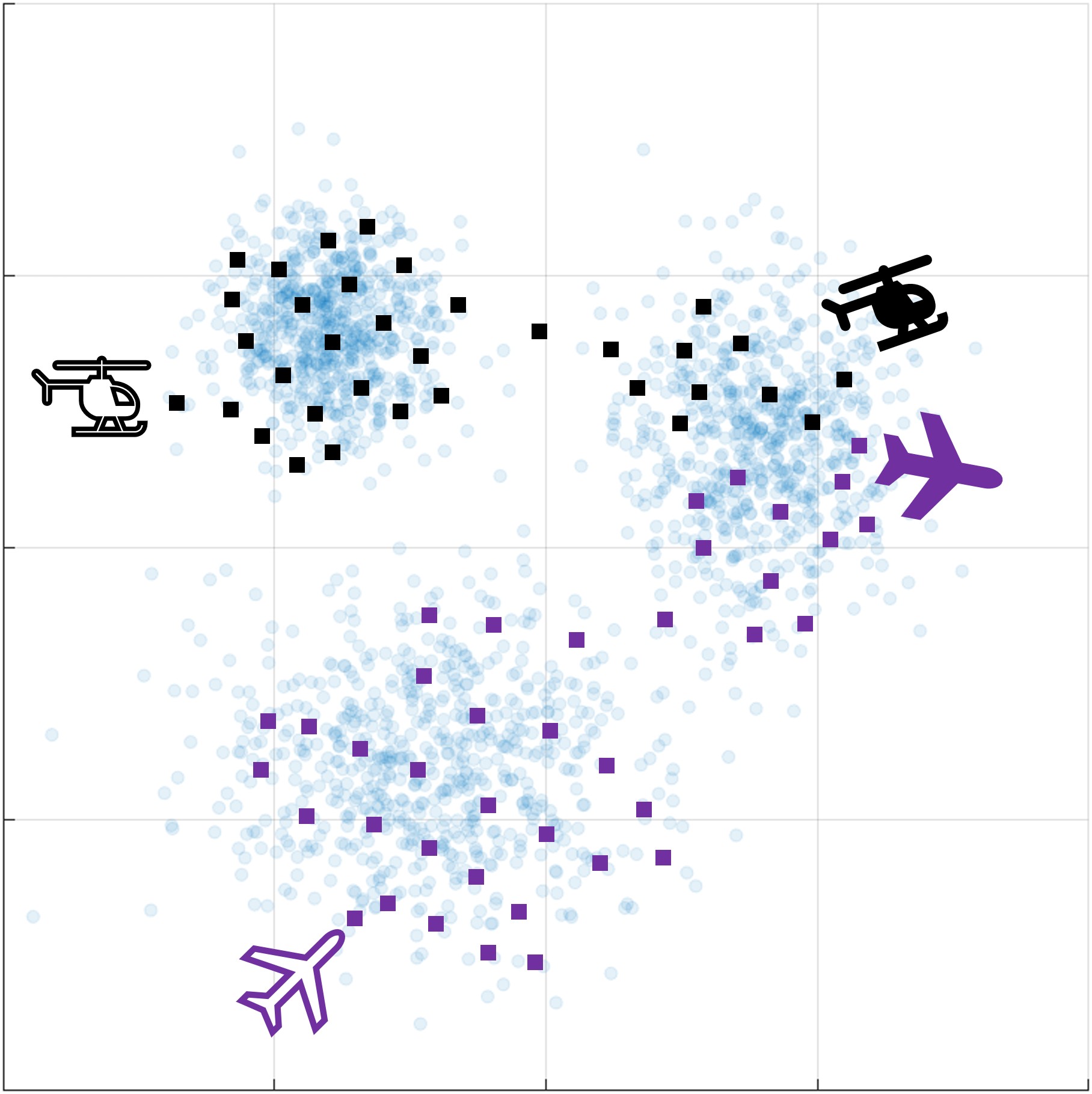}}
    \caption{Conceptual illustration of prioritized search coverage using a multi-agent system.}
    \label{fig: concept}
\end{figure}
These reference points indicate high-priority regions within the disaster zone. By directing search agents toward these areas, the detection mission becomes significantly more efficient compared to a uniform search strategy. Specifically, the objective is for agents to systematically and collaboratively cover the reference points, ensuring their trajectories (black and purple squares in Fig.~\ref{fig: concept}(b)) closely match the reference distribution (blue circles in Fig.~\ref{fig: concept}(a)). This approach optimally allocates resources, balancing constraints such as agent dynamics, control input limitations, and battery life, ultimately improving the effectiveness of victim detection efforts in large-scale search and rescue (SAR) missions.
To address this problem, optimal transport theory, which is explained in detail in the next section, is introduced.

\subsection{Preliminary: Optimal Transport Theory}
Optimal transport (OT) is a field focused on minimizing the cost of moving (or transporting) mass from one distribution to another \cite{monge1781memoire}, which can be mathematically explained as follows.

Consider two mass distributions: 
$\mu=\{(\y_i,\alpha_i)\ |\ i=1,2,...,M,\ \y_i\in\mathbb{R}^2,\ \alpha_i\in\mathbb{R}_{\geq0}\}$ and
$\nu = \{(q_j,\beta_j) \ |\ j=1,2,...,\ N, \ q_j\in\mathbb{R}^2,\ \beta_j\in\mathbb{R}_{\geq0}\}$, where $\y_i$ and $q_j$ represent the location of the points in each mass distribution, and $\alpha_i$ and $\beta_j$ denote their corresponding weight. The task is to move the mass in $\nu$ such that the resulting distribution matches $\mu$. The OT problem aims to determine the most efficient way to achieve this transportation and can be formulated as {\small
\begin{align}
&\min_{\gamma_{ij}\geq 0}\sum\nolimits_{i,j}\gamma_{ij}C(\y_{i},q_{j})\label{eqn: W_LP} 
		\\&\text{subject to} \  \left\{
  \begin{aligned}
  &\sum\nolimits_{i=1}^{M}\gamma_{ij} = \beta_{j}, \, \forall j, \quad \sum\nolimits_{j=1}^{N}\gamma_{ij} =\alpha_{i}, \, \forall i,\\
    &\sum\nolimits_{i=1}^{M}\alpha_i = \sum\nolimits_{j=1}^{N}\beta_j = 1,
    \end{aligned}
 \right.\nonumber
\end{align}}
where $\gamma_{ij}$ is the decision variable representing the amount of mass transported from the point $q_{j}$ to $\y_{i}$, and $C(\y_{i},q_{j})$ is the cost for the transportation.

The first and second constraints in \eqref{eqn: W_LP} imply that all the mass $\beta_j$ must be fully transported from the point $q_j$ without any remainder, and the transported mass should accumulate to form $\alpha_i$ at $\y_i$. The last constraint ensures that the total mass of distribution $\mu$ equals the total mass of distribution $\nu$ due to mass conservation, meaning that mass cannot be created or lost during transportation. For probabilistic distributions, total mass (or probability) must equal one. 

A $p$-th power of the Euclidean distance can be used for the cost function $C(\y_{i},q_{j})$, leading to the $p$-Wasserstein distance \cite{villani2009optimal} as follows:
{\small
\begin{equation}\label{eqn: W_LP2}
\begin{aligned}
& \mathcal{W}_p(\mu, \nu) =
 \min_{\gamma_{ij} \geq 0} (\sum\nolimits_{i,j}\gamma_{ij}\lVert \y_i - q_j \rVert^p)^{1/p}
\end{aligned}
\end{equation}}
with the same constraints in \eqref{eqn: W_LP}.

While the case $p=1$ for $\mathcal{W}_{1}$ (commonly known as the Earth Mover’s Distance) could be considered, the 2-Wasserstein distance $\mathcal{W}_{2}$ is chosen due to its favorable mathematical properties, including its ability to capture second-order moment information, which is essential for the problem at hand. For clarity, the notation $\mathcal{W}$ will be used hereafter to refer specifically to the 2-Wasserstein distance.

\subsection{Problem-Solving Strategy}
To achieve the similarity between the distribution of agent trajectories and the distribution of reference points, OT theory is employed. 
In our problem setup, two mass distributions are defined over a two-dimensional space: the distribution of agent trajectories and the distribution of reference points.

\textbf{The distribution of reference points, $\boldsymbol{\nu}$.} The distribution of reference points is defined as $\nu = \{(q_j,\beta_j) \ |\ j=1,2,...,N,\ q_{j}\in \mathbb{R}^2,\ \beta_j\in\mathbb{R}_{\geq 0}\}$. This distribution corresponds to the reference point distribution shown in Fig. \ref{fig: concept}(a). The symbol $q_{j}$ represents the location of sample point $j$ in the distribution, and $\beta_j$ denotes its corresponding mass. The mass  $\beta_j$ must sum to 1 over all $j$, satisfying the constraint of the OT theory in \eqref{eqn: W_LP}. The sample point distribution $\nu$ is assumed to be given in advance, constructed from multiple pre-acquired sources. If the source is given as a continuous probability density function, it can be approximated as a discrete mass distribution using techniques such as a random sampling method.

\textbf{The distribution of agent trajectories, $\boldsymbol{\mu}$.} The distribution of agent trajectories is defined as $\mu=\{({}^r\y^{k},{}^r\alpha^{k})\ |\ r=1,2,..,L,\ k=1,2,...,{}^r M,\  {}^r\y^{k}\in\mathbb{R}^2,\ {}^r\alpha^{k}\in\mathbb{R}_{\geq 0}\}$. This distribution corresponds to the set of agent trajectories shown in Fig. \ref{fig: concept}(b). Similar to the distribution of reference points $\nu$, the symbol ${}^r\y^{k}$ denotes the location of agent $r$ at discrete time $k$, which ranges from the initial time, $k=1$, to the estimated (nominal or designed) operation time of agent $r$, ${}^r M$, potentially determined by its battery life or user settings. This is analogous to $\{\y_i\}$ in the previous section. Correspondingly, the mass of this point is assigned as ${}^r\alpha^k$. The mass ${}^r\alpha^k$ is assigned based on the energy consumption of agent $r$ at time $k$. In this paper, agents are assumed to consume the same amount of energy at every time instance, resulting in equal mass, ${}^r\alpha^k = 1/(\sum_r {}^r M)$, satisfying $\sum_r \sum_k {}^r\alpha^k = 1$, as given in \eqref{eqn: W_LP}. 

For clarity, points in $\mu$ and $\nu$ will be referred to as \textit{agent-point} and \textit{sample-point}, respectively, for the remainder of this paper. Furthermore, the term \textit{weight} indicates the mass of each agent-point and sample-point.
\begin{remark}[Notation Conventions for Variable Indices]
    Throughout this paper, $j$ will be used to index the sample-points, $r$ will be used to index the agents, and $k$ will be used to index the discrete-time steps. Furthermore, each index has a specific position in the notation: the index $r$ appears as a left superscript, $j$ as a right subscript, and $k$ as a right superscript.
\end{remark}

Based on the defined notations, the optimal area coverage is achieved by finding the control inputs that minimize the 2-Wasserstein distance between $\mu$ and $\nu$, which is formulated by
{\small
\begin{align}
    &\mathbf{u}^* = \argmin\nolimits_{\mathbf{u}} \w(\mu,\nu)\label{eqn: general formulation}\\
    &\text{subject to} \  {}^r\y^{k+1} = f({}^r\y^{k}, {}^r u^{k}),\,C_u\,{}^ru^k \leq D_u,\,\forall k \in \mathbb{N}_{1:{}^r M}\nonumber
\end{align}}
where $\mathbf{u}=\{{}^ru^k\ |\ r=1,2,...,L,\ k = 1,2,...,{}^rM\}$ denotes a set of control inputs, and $\mathbf{u}^*$ represents a set of optimal control inputs. The first constraint represents agent dynamics, while the second constraint specifies the input threshold. Additionally, the terminal time is limited by the estimated operation time. These constraints reflect the practical limitations.

This minimization problem involves high-dimensional decision variables and highly complex constraints, making it computationally demanding and requiring a nonlinear solver, which may be impractical. To mitigate these challenges, instead of considering all agent-points and sample-points simultaneously, a subset of sample-points—referred to as local sample-points—and a single agent-point are considered at each discrete-time step. A suboptimal solution is then computed to minimize the Wasserstein distance between them. 
In this case, the progress of area coverage is updated by transporting some weights of sample-points to an agent-point at each time step, leading to a time-varying weight for each sample-point, denoted as $\beta^k_j$.  

In a multi-agent system, agents may cover different areas, and sample-points may transfer their weights to different agents. Collaborative coverage is achieved through weight sharing between agents; however, before communication occurs, each agent processes the weights of sample-points independently. To indicate this, the weight of the sample-point maintained by agent $r$ is denoted as ${}^r\beta^k_j$ throughout this paper.

In particular, this paper is based on the following considerations:
\begin{itemize}
    \item[\textbullet] Although the general formulation in (3) allows for nonlinear systems, this paper focuses on linear time-invariant (LTI) systems.
    \item[\textbullet] 
    While the proposed method supports a decentralized communication setting, the simulations in this study assume all-to-all communication in order to focus solely on evaluating the controller's coverage performance.
\end{itemize}

\section{Density-based Predictive Control (DPC)}
To address the area coverage problem described in the previous section, Density-based Predictive Control (DPC) is proposed. The DPC scheme consists of three stages: the optimal control stage, the weight update stage, and the weight-sharing stage. These three stages alternate at each time step until the entire weight in the sample-point distribution $\nu$ is transported to the agent-point distribution $\mu$. In the optimal control stage, a subset of sample-points from the entire set is selected as the local sample-points. Based on these local sample-points, the control input is computed to minimize the cost function associated with the Wasserstein distance between local sample-points and a future agent-point. After the control input is applied and the agent moves to a new location at the next time step, the weight of the sample-points surrounding the agent is transported to the agent-point during the weight update stage. As a result, the weight of those sample-points decreases. This process reflects areas visited by updating the weight of the sample-points within those areas. Finally, agents share the weight of sample-points they hold with each other during the weight-sharing stage. The system description and each stage of DPC are explained in detail.

\subsection{General Linear Time-Invariant (LTI) System Description}
Consider the dynamics of an LTI system for agent $r$ as follows:
{\small    \begin{align}
\begin{aligned}
      &{}^r\mathsf{x}^{k+1} = {}^r A\ {}^r\xx^{k} + {}^r B\ {}^ru^k,\ {}^r\y^{k} = {}^r C\ {}^r\xx^k,  
\end{aligned}\label{eqn: LTI system}
\end{align}}
where ${}^r\mathsf{x}^k\in\mathbb{R}^{n}$ is a state vector, ${}^ru^k\in\mathbb{R}^{m}$ is an input vector, and ${}^r\y^k\in\mathbb{R}^{p}$ is an output vector, indicating the Cartesian spatial coordinates of the agent. The system matrices are given by ${}^r A\in\mathbb{R}^{n\times n}$, ${}^r B\in\mathbb{R}^{n\times m}$ and ${}^r C\in\mathbb{R}^{p\times n}$ , where the matrices ${}^r B$ and ${}^r C$ have ranks $m$ and $p$, respectively. The index $r$ in the system matrices signifies that this controller is designed for a heterogeneous multi-agent system. Since the proposed control strategy applies to each individual agent within a fully decentralized framework, the left superscript \( r \), which indicates the agent index, will be omitted hereafter. Collaborative area coverage and weight sharing among agents are discussed in Section 3.4.

\begin{remark}
For the LTI system given in \eqref{eqn: LTI system}, there exists an output relative degree $P\in \mathbb{N}$ such that the term $CA^{P-1}B$ is the first non-zero matrix in the sequence starting from the matrix $CA^{0}B$. In other words, $CA^{0}B = CA^{1}B =\cdots= CA^{P-2}B = \mathbf{0}$.
\label{Assumption: cancellation}
\end{remark}
\begin{example}\label{example1}
    Cancellations in the product of system matrices, such as $CA^{0}B = CA^{1}B =\cdots= CA^{P-2}B = \mathbf{0}$, frequently occur in real-world systems. For instance, a simple system $\ddot{x}=u$ is approximated to the discrete-time state-space equation using Euler's method by 
    {\small \begin{align}
        \underbrace{\begin{bmatrix}
            {x}^{k+1}\\ {v}^{k+1}
        \end{bmatrix}}_{\xx^{k+1}} = \underbrace{\begin{bmatrix}
            1& \Delta T\\0 &1
        \end{bmatrix}}_{A}
        \underbrace{\begin{bmatrix}
            x^{k}\\ v^{k}
        \end{bmatrix}}_{\xx^{k}}+
        \underbrace{\begin{bmatrix}
            0\\ \Delta T
        \end{bmatrix}}_{B} u^k,\quad \y^k=\underbrace{\begin{bmatrix}
            1&0
        \end{bmatrix}}_{C}\underbrace{\begin{bmatrix}
            x^{k}\\ v^{k}
        \end{bmatrix}}_{\xx^k}.\nonumber
    \end{align}}
    In this system, $CB$ is zero, indicating the control input $u^k$ does not directly affect $\y^{k+1}$ since $\y^{k+1} = C\xx^{k+1} = CA\xx^{k}+\cancelto{0}{CB}u^{k}$. However, $CAB$ is a nonzero matrix, indicating the control input $u^k$ does affect $\y^{k+2}$. Therefore, the output relative degree $P$ of this system is 2. Similarly, for a discrete-time state-space equation of  ${x}^{(4)}=u$, the matrix products $CA^2B=CAB=CB=0$, whereby $P$ is 4. The value $P$ represents the discrete-time interval after which $u^k$ first affects the system's output.
\end{example}
The output relative degree $P$ is utilized to formulate the optimal control input of the DPC scheme. In what follows, we explain each stage of the DPC scheme.

\subsection{Stage A: Optimal Control Stage}
The control scheme aims to achieve similarity between the sample-point distribution and the agent-point distribution by minimizing the Wasserstein distance between the local sample-points and the agent's estimated location at $P$ time steps ahead at each time step. The agent's estimated location at $P$ time steps ahead is considered since the current control input first affects the output at that time, as mentioned in Remark \ref{Assumption: cancellation} and Example \ref{example1}. Therefore, the control input is calculated based on the predictive ($P$ time steps ahead) Wasserstein distance.     

\textbf{Local sample-points. }
Here, we first describe the selection process of local sample-points at each time step. At any given time, the current agent-point possesses a pre-assigned weight, $\alpha^k$, as previously explained. According to OT theory, a certain weight from multiple sample-points must be transported to this agent-point while ensuring that the total transported weight matches the weight of the agent-point. This transported weight from the local sample-points is denoted as $\bar{\beta}_j^k$, which may differ from the current weight of the sample-points, $\beta_j^k$.

The key challenge lies in determining both the amount of weight to be transported and the specific sample-points (with non-zero weights) from which it should be transported to the current agent-point. To address this, we introduce the weight-normalized Euclidean (wnE) distance at time $k$, defined as: 
{\small
\begin{align}
    d_{\text{wnE},j}^k = \norm{q_{j}-\bar{q}^{k-1}}/{\beta^k_{j}}, \ \forall j:\ \beta_j^k > 0,\label{eqn:priority_index}
\end{align}}
where $\bar{q}^{k-1}$ represents the mass center of the local sample-points from the previous time step, computed as their weighted average. Specifically, for $k=1$, it coincides with the initial agent-point, i.e., $\bar{q}^{0} = \y^0$, while for $k>1$, it is determined using the set $\s^{k-1}$, which consists of the local sample-points at time $k-1$, as  
{\small\begin{align}
    \bar{q}^{k-1} = (\sum\nolimits_{q_j\in \s^{k-1}} \bar{\beta}^{k-1}_j  q_j)/(\sum\nolimits_{q_j\in \s^{k-1}} \bar{\beta}^{k-1}_j).\label{eqn: mass_center_LSP}
\end{align}}
\noindent In this way, the mass center, $\bar{q}^{k-1}$, is predetermined at each time $k$. Then, the selection of new local sample-points follows an ascending order of $d_{\text{wnE}}$ until the sum of the selected local sample-points' weights matches the weight of the agent-point, i.e., $\sum\nolimits_{j} \bar{\beta}^{k}_j = \alpha^{k}$. This weight condition is essential to preserve mass conservation in OT theory.

\begin{remark}[Interpretation of the mass center, $\bar{q}^{k}$]
    The weight-normalized Euclidean distance \eqref{eqn:priority_index} serves as the selection metric for local sample-points. In this metric, $\bar{q}^{k}$ turns out to be the desired location of the agent-point at time $k+P$, as described in the following discussion. 
    After the local sample-points are selected, the objective of this controller is to minimize the Wasserstein distance between the local sample-points at the current time $k$ and the agent-point at $k+P$, formulated as
    {\small\begin{align*}
        &\min_{\gamma^{k}_j,\  \y^{k+P}} \sum\nolimits_{q_j\in \s^{k}} \gamma^{k}_j\norm{\y^{k+P}- q_j}^2,
        \\& \text{subject to} \quad \gamma_j^{k} = \bar{\beta}_j^{k},\ \forall j,\ \sum\nolimits_{q_j\in\s^{k}} \gamma_j^{k} = \alpha^{k},
    \end{align*}}
    where $\gamma^{k}_{j}$ represents the transportation plan. The first constraint is obtained by substituting $M=1$ into the constraint given in \eqref{eqn: W_LP}. 
    Then the optimizer of $\y^{k+P}$ in this minimization problem has a closed-form solution, given by $\bar{q}^{k}$. The derivation process is straightforward and omitted in this paper. This result implies that the mass center $\bar{q}^{k}$ serves as the desired location of the agent at time $k+P$. 
\end{remark}

Based on $\s^k$, the predictive Wasserstein distance $(\w^{k+i|k})^2$ between $\s^k$ and the agent's predicted location at $i$ steps ahead is calculated by
{\small\begin{align}
   (\mathcal{W}^{k+i|k})^2 = \sum\nolimits_{q_j\in \s^{k}} \bar{\beta}^k_j \norm{\y^{k+i}-q_j} ^2 .\label{eqn: predictive WD_1}
\end{align}}
According to Remark \ref{Assumption: cancellation}, there is a value $P$, after which the agent's location is first affected by the current control input. Since the Wasserstein distance is a function of location, the control input at time $k$ does not affect $\w^{k+1|k}$, $\cdots$, $\w^{k+P-1|k}$ but rather affects $\w^{k+P|k}$. Based on this analysis, a function for determining convergence in terms of the Wasserstein distance is formulated as 
{\small\begin{align}
    &\Delta \w^k = (\mathcal{W}^{k+P|k})^2 - (\mathcal{W}^{k|k})^2 \label{eqn: diff_Wass} 
    \\&= \sum\nolimits_{q_j\in \s^{k}} \bar{\beta}^k_j \norm{\y^{k+P}-q_j} ^2 - \sum\nolimits_{q_j\in\s^{k}} \bar{\beta}^k_j \norm{\y^{k}-q_j}^2.\nonumber
\end{align}}
The convergence condition of the Wasserstein distance is to show that $\w^{k+P}$ decreases compared to $\w^{k}$ (i.e., $\Delta \w^k$ is negative).
The function $\Delta \w^k$ is simplified as stated in the following proposition.
\begin{proposition}[Simplification of $\Delta \w^k$]\label{prop: simplification of Dwk}
    The function $\Delta \w^k$ in \eqref{eqn: diff_Wass} is simplified by
    {\small\begin{align}
        &\Delta \w^k=\alpha^k\left(\norm{\y^{k+P}- \bar{q}^k}^2-\norm{\y^{k}-\bar{q}^k}^2\right). \label{eqn: diff_Wass_2}
    \end{align}}
\end{proposition}
\begin{proof}
    The detailed proof is provided in Appendix A.
\end{proof}

Fig. \ref{fig: conv_cond} visually represents the equation \eqref{eqn: diff_Wass_2}. The orange dashed line represents the circle with the radius of $\norm{\y^{k}-\bar{q}^k}$, centered at $\bar{q}^k$. In the figure, the sign of $\Delta \w^k$ is determined based on the location of $\y^{k+P}$. If $\y^{k+P}$ is inside the circle, $\Delta \w^k$ is negative. If $\y^{k+P}$ lies on the circle, $\Delta \w^k$ equals zero. Otherwise, $\Delta \w^k$ becomes positive.
\begin{figure}[h]
    \centering
    \includegraphics[width=0.6\linewidth]{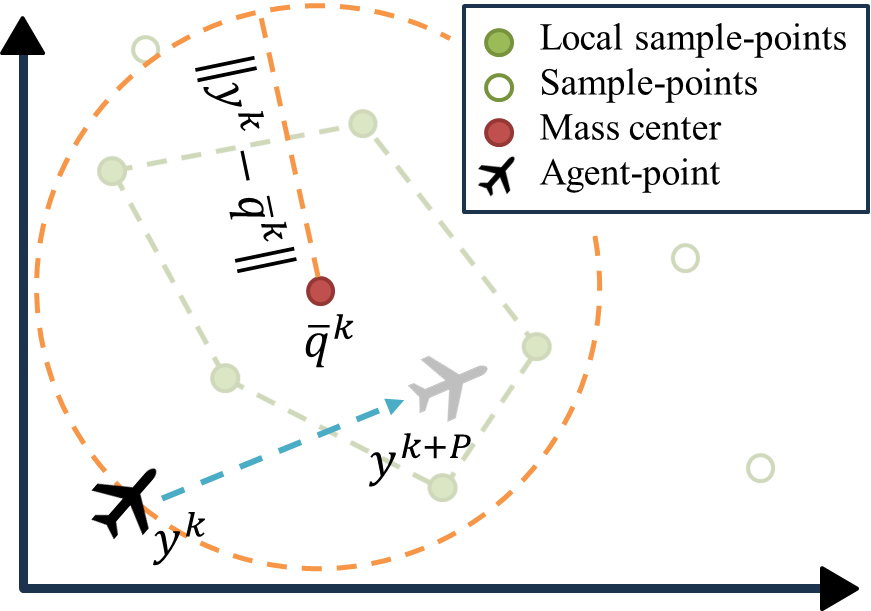}
    \caption{Description of the function $\Delta \w^k$.}
    \label{fig: conv_cond}
\end{figure}

In \eqref{eqn: diff_Wass_2}, by substituting \eqref{eqn: LTI system}, $\y^{k+P} = C\xx^{k+P} = CA^P\xx^k+CA^{P-1}Bu^k+\cdots+CBu^{k+P-1}$. Since $CA^{0}B = CA^{1}B =\cdots= CA^{P-2}B = \mathbf{0}$, as stated in Remark \ref{Assumption: cancellation}, the equation simplifies to $\y^{k+P} = CA^P\xx^k+CA^{P-1}Bu^k$. Substituting this $\y^{k+P}$ to \eqref{eqn: diff_Wass_2}, 
{\small\allowdisplaybreaks
\begin{align}
    &\Delta \w^k(u^k) = (u^k)^\top \D_1 u^k + 2\D_2 u^k + \D_3,
    \label{eqn: diff_Wass_2.2}
\\&\text{where}\quad    \D_1 = \alpha^k B^\top (A^{P-1})^\top C^\top C A^{P-1}B ,\nonumber
    \\&\phantom{\text{where}\quad}    \D_2 = \alpha^k\left\{(\xx^k)^\top (A^P)^\top C^\top-(\bar{q}^k)^\top \right\}C A^{P-1}B\nonumber ,
    \\&\phantom{\text{where}\quad}\D_3 = \alpha^k(\xx^k)^\top\left\{(A^P)^\top C^\top C A^P - C^\top C \right\}\xx^k \nonumber
    \\&\phantom{\text{where}\quad}\qquad \quad -2 \alpha^k(\bar{q}^k)^\top C (A^P-I_n)\xx^k. \label{eqn: D1 to D3}
\end{align}
}
Then the optimal control input $({u^*})^k$ and the convergence range are derived as follows.
\begin{theorem} [The convergence condition and the optimality of DPC -- Unconstrained case]\label{thm: 1}
    For the LTI system \eqref{eqn: LTI system}, if $CA^{P-1}B$ has a rank of $p\leq m$, where $m$ is the dimension of the control input $u^k$, then the optimal control inputs are a set-valued solution derived with an arbitrary vector $h\in\mathbb{R}^m$ by     {\small\begin{align}        ({u^*})^k = -\D_1^+\D_2^\top + (I_m-\D_1^+\D_1)h, \label{eqn: opt_cont_input}    \end{align}}    and the convergence range of DPC with respect to the control input $u^k$ is obtained by
    {\small
    \begin{align}    \norm{u^k+\D_1^+\D_2^\top}^2_{\D_1}<\D_2 \D_1^+ \D_2^\top-\D_3. \label{eqn: general_conv_range}    \end{align}}
    Specifically, if $CA^{P-1}B$ has a rank of $p = m$, then the unique optimal control input of DPC is given by{\small\begin{align}        ({u^*})^k = -\D_1^{-1}\D_2^\top, \label{eqn: opt_cont_input_unique}    \end{align}}    for which the convergence range of DPC with respect to the control input $u^k$ is obtained by     {\small\begin{align}    \norm{u^k+\D_1^{-1}\D_2^\top}^2_{\D_1}<\D_2 \D_1^{-1} \D_2^\top-\D_3. \label{eqn: conv_range}    \end{align} }
\end{theorem}

\begin{proof}
The detailed proof is provided in Appendix B.    
\end{proof}

\begin{remark}[Physical interpretation of the convergence range]
In \eqref{eqn: general_conv_range}, since the matrix $\D_1$ is positive semidefinite, $\norm{u^k+\D_1^{-1}\D_2^\top}^2_{\D_1}\geq 0$. Thus, if $\D_2 \D_1^+ \D_2^\top-\D_3$ is negative, then no input $u^k$ exists satisfying \eqref{eqn: general_conv_range}, meaning that the convergence range is empty ($\emptyset$). On the other hand, if $\D_2 \D_1^+ \D_2^\top-\D_3$ is nonnegative, the convergence range consists of the union of the column space of $I_m-\D_1^+ \D_1$ and a subset of the column space of $\D_1$ that satisfies \eqref{eqn: general_conv_range}, forming a convex shape. For example, if the dimension of $\D_1$ is 2, the subset forms an ellipse. If the dimension of $\D_1$ is 3, then the subset forms an ellipsoid.
\end{remark}

If one chooses a control input $u^k$ satisfying the convergence condition \eqref{eqn: general_conv_range}, then it is guaranteed that the local Wasserstein distance, $\w^k$, denoted by the Wasserstein distance between the agent's position and $q_j\in\s^k$, converges in a piecewise manner.

\begin{remark}
The local Wasserstein distance is described as \textit{piecewise} convergent because, in \eqref{eqn: diff_Wass}, its difference is formulated over a horizon window of length \( P+1 \) using selected local sample-points. This implies that convergence is only ensured within this horizon and not beyond it. Moreover, the Wasserstein distance exhibits piecewise continuity, as new local sample-points are selected at each discrete-time step.
\end{remark}

The analytic optimal solution is given by \eqref{eqn: opt_cont_input} when the system is not subject to input constraints. For the constrained case, where the control input thresholds exist, such as polyhedral constraints $C_u u^k \leq D_u$, the optimal control input $({u^*})^k$ can be found numerically using quadratic programming, where the problem is given by
{\small\begin{align}
\begin{aligned}
    &({u^*})^k = \argmin\nolimits_{u^k} \ (u^k)^\top \D_1 u^k + 2\D_2 u^k + \D_3
    \\&\text{subject to} \quad C_uu^k \leq D_u.
    \end{aligned}\label{eqn: QP_opt_cont_input}
\end{align}}
The entire process of Stage A is summarized as the pseudocode shown in Algorithm \ref{Alg: StageA}.

\begin{algorithm}[!b]
    \caption{Stage A--Optimal Control Stage}
    \label{Alg: StageA}
    {
        \small
        \SetKwInOut{Input}{Input}
        \SetKwInOut{Output}{Output}
        \Input{
            $\alpha^k$, $\beta^k_j, q_j, \forall j=1,2,...,N$, $\bar{q}^{k-1}$, $\xx^k$, and system matrices $A$, $B$, $C$
        }

    $\s_{>0}
     \gets \{\, j\in \{1,\dots,N\} \mid \beta^k_j>0\,\}$
    
    $\alpha_\text{rem}\gets \alpha^k$

    $\s^{k}\gets \emptyset$
    
    \For{$i\in\s_{>0}$}{
    $d_{\text{wnE},i}^k \gets \norm{q_{i}-\bar{q}^{k-1}}/{\beta^k_{i}}\ $ \texttt{// Eq.} \eqref{eqn:priority_index}
    }

    $\s_{>0}$ $\gets$ Sort($\s_{>0}$) by $d_{\text{wnE},i}^k$ in ascending order

    \tcp{Local sample-points selection}
    \For{$i \in \s_{>0}$}{
      $\bar{\beta}^{k}_i \gets \min\!\big(\alpha_{\mathrm{rem}},\ {}^{r}\beta^{k}_i\big)$\\
      $\alpha_{\mathrm{rem}} \gets \alpha_{\mathrm{rem}} - \bar{\beta}^{k}_i$\\
      $\s^{k} \gets \s^{k} \cup \{(q_i, \bar{\beta}^{k}_i)\}$\\
      \If{$\alpha_{\mathrm{rem}} = 0$}{\textbf{break}}
    }

    \tcp{Optimal control input computation}
    $\bar{q}^{k} \gets (\sum\nolimits_{q_j\in \s^{k}} \bar{\beta}^{k}_j  q_j)/(\sum\nolimits_{q_j\in \s^{k}} \bar{\beta}^{k}_j)$ \texttt{// Eq.} \eqref{eqn: mass_center_LSP}

    $\D_1,\D_2,\D_3 \gets \text{computed as in \eqref{eqn: D1 to D3}}$
    
    \If{$u^k$ is unconstrained}{
    $({u^*})^k = -\D_1^+\D_2^\top + (I_m-\D_1^+\D_1)h$ 
    
    \qquad\qquad\qquad\qquad\qquad\qquad\qquad\texttt{// Eq.} \eqref{eqn: opt_cont_input}}
    \Else{$\begin{aligned}
    &({u^*})^k = \argmin\nolimits_{u^k} \ (u^k)^\top \D_1 u^k + 2\D_2 u^k + \D_3
    \\&\text{subject to} \quad C_uu^k \leq D_u\  \texttt{// QP; Eq.}\ \eqref{eqn: QP_opt_cont_input}
    \end{aligned}$} 
}
\end{algorithm}

Once the control input is found for DPC while satisfying the convergence condition, it can be applied to the agent to move to a new location, followed by Stage B for weight update of sample-points.

 \subsection{Stage B: Weight Update Stage}
After agent $r$ moves from the previous position ${}^r\y^k$ to the next position ${}^r\y^{k+1}$, the mass of sample-points around ${}^r \y^{k+1}$ should be transported to this agent-point to update the coverage progress while satisfying the mass conservation law described by the last constraint in \eqref{eqn: W_LP}. 
By this transportation, the mass of the sample-points decreases, which, in turn, leads to a lower weight-normalized Euclidean distance \eqref{eqn:priority_index} for the unvisited sample-points. This stage is conceptually illustrated in Fig. \ref{fig: concept_WeightUpdate}.
\begin{figure}[h]
    \centering
    \subfloat[Before the update]{
    \includegraphics[width=0.38\linewidth]{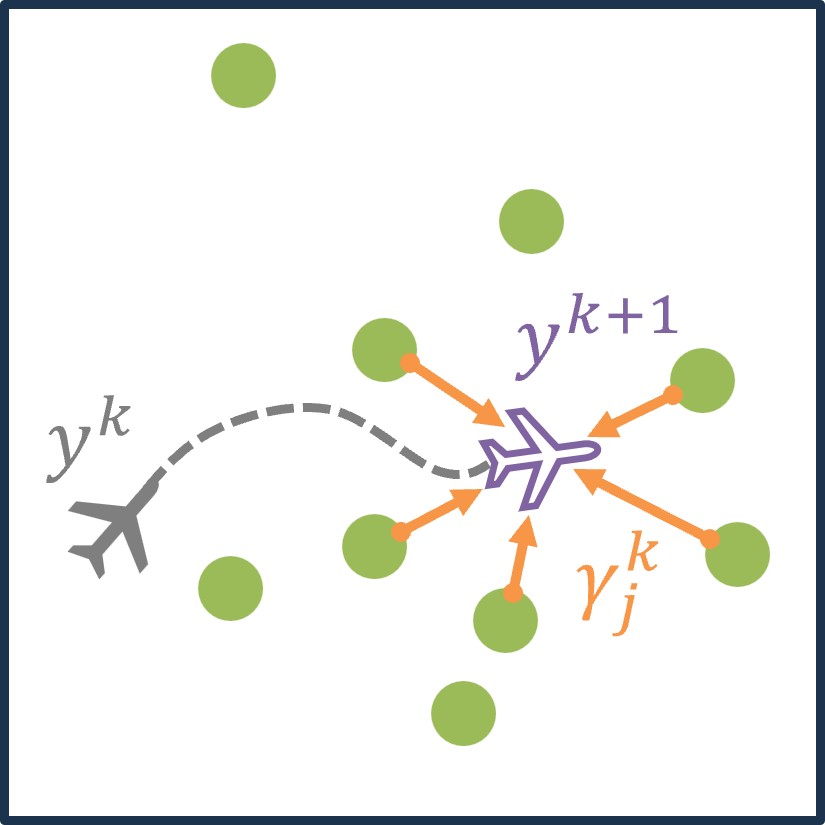}}\quad
    \subfloat[After the update]{
    \includegraphics[width=0.38\linewidth]{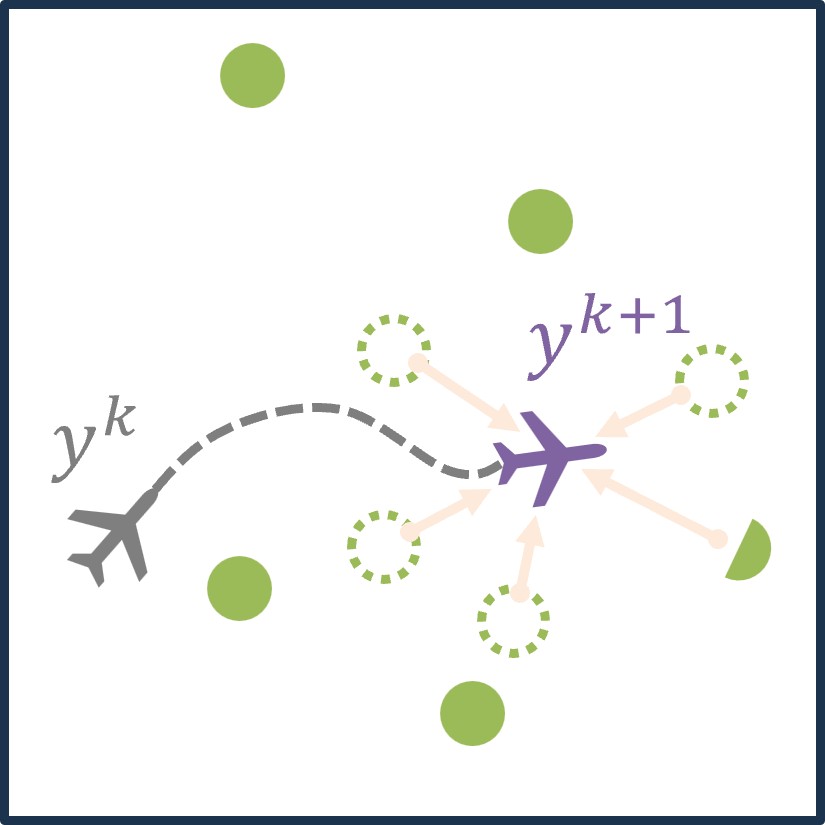}}
    \caption{Conceptual illustration of the weight update stage. Green circles represent the sample-points, while the dashed circles represent sample-points with no weight after transportation.}
    \label{fig: concept_WeightUpdate}
\end{figure}

The amount of the weight to be transported from each sample-point to the agent-point is determined by a minimization problem formulated by
{\small
\begin{align}
&{}^r ({\gamma}_{j}^*)^{k}=\argmin\nolimits_{{}^r \gamma_{j}^{k}} \sum\nolimits_{j}{}^r \gamma_{j}^{k}\lVert {}^r\y^{k+1} - q_j \rVert^2 \label{eqn: Sec3_WeightUpdate_minimization} \\
		&\text{subject to} \ \ 
  0 \leq {}^r{\gamma}^{k}_{j} \leq  {}^r{\beta}^k_{j},\ 
		\sum\nolimits_{j=1}^{N}{}^r{\gamma}^{k}_{j} = {}^r \alpha^{k+1}, \  \forall j,\nonumber
\end{align}}
where ${}^r \gamma_{j}^{k}$ is the transportation plan from sample-point $j$ at time $k$, and ${}^r \alpha^{k+1}$ is the weight of ${}^r \y^{k+1}$. The first constraint denotes that the transportation plan must be nonnegative and that each sample-point cannot transport more weight than it currently possesses. The second constraint is to obey the mass conservation law. Since \eqref{eqn: Sec3_WeightUpdate_minimization} is a linear programming (LP) problem, the solution can be easily obtained. Also, an analytic solution for this problem is well investigated in \cite{kabir2021wildlife}.
The weight of each sample-point held by agent $r$ is then updated by
{\small\begin{align}
	{}^r\beta^{k+1}_{j} = {}^r{\beta}^{k}_{j}  -  {}^r({\gamma}_{j}^*)^{k},\  \forall j.\label{eqn: weight update}
\end{align}}
\subsection{Stage C: Weight-Sharing Stage}

Stages A and B are executed independently by each agent under a decentralized framework, resulting in non-collaborative coverage as agents are unaware of areas already covered by others. To address this, Stage C introduces a weight-sharing mechanism where agents within communication range exchange coverage information to enable local coordination. 

Each agent maintains local weight data for sample points, denoted by ${}^r\beta^{k}_{j}$. When agents \( r \) and \( s \) are within range \( d_{\text{comm.}} \), they synchronize their weights using:
{\small
\begin{align}\label{eqn: D2C weight update}
    {{}^r\beta}^{k}_{j} =  {^s\beta}^{k}_{j} = \min({^r\beta}^{k}_{j},\,{^s\beta}^{k}_{j}),\quad \forall j.
\end{align}}
This local synchronization helps prevent redundant coverage by aligning agents’ understanding of which areas have already been visited. In contrast, the centralized case assumes local controllers with centralized communication, where each agent retains its autonomy but has access to globally synchronized weight information. This corresponds to the idealized case where \( d_{\text{comm.}} \to \infty \).

Thus, collaboration in the decentralized setup is achieved through local updates during agent encounters, supporting efficient area coverage. Additional details are available in \cite{kabir2021wildlife}.

\begin{remark}[Scalability of the DPC scheme]\label{Rem: Scalability}
    Since the proposed DPC scheme is designed for multi-agent systems operating over large areas, scalability is crucial for practical implementation. To this end, we have analyzed the theoretical runtime complexity of each stage. The symbols $L$ and $N$ are used to represent the number of agents and the number of sample-points, respectively.

        \textbf{Stage A:} Owing to the fully decentralized framework, each agent $r=1, ..., L$ executes this stage individually only using its local weight information, ${}^r\beta^k_j$ for $j=1,...,N$. The primary steps in this stage are the selection of local sample-points and the computation of the optimal control inputs, as shown in \eqref{eqn: opt_cont_input}, \eqref{eqn: opt_cont_input_unique}, and \eqref{eqn: QP_opt_cont_input}. Considering that the dimension of the input vector, $m$, is typically much smaller than $N$, local sample-points selection dominates the runtime of Stage A. Consequently, the primary bottleneck, namely the sorting of the sample-points with $O(N\log N)$ complexity, governs the overall computational cost of Stage A. 
        
        \textbf{Stage B:} Similar to Stage A, each agent individually updates its local weight information. Although the main computation in \eqref{eqn: Sec3_WeightUpdate_minimization} is formulated as a linear program, an analytic solution can be obtained by sorting the sample-points in ascending order of $\lVert {}^r\y^{k+1} - q_j \rVert^2$ \cite{kabir2021wildlife}. Thus, the theoretical complexity of this stage is dominated by $O(N\log N)$.
        
        \textbf{Stage C:} In case of all-to-all communication, the theoretical maximum communication overhead is computed as $O(L\times(L-1)\times N)$. However, the range-limited communication strategy in Stage C, together with the event-triggered communication strategy, can effectively reduce the communication overhead.     
\end{remark}

\section{Simulations}
Simulations are conducted to validate the proposed DPC scheme. In the simulations, the first-order dynamics and the linearized quadrotor dynamics are considered. The sample-point distribution used in the simulations, shown in Fig. \ref{fig: refer_dist}, is modeled as a mixture of Gaussian distributions. The sample-point representation in Fig. \ref{fig: refer_dist}(b) is obtained by sampling point clouds from this mixture.
\begin{figure}[h]
    \centering
    \subfloat[Contour map of the reference distribution]{
    \includegraphics[width=0.45\linewidth]{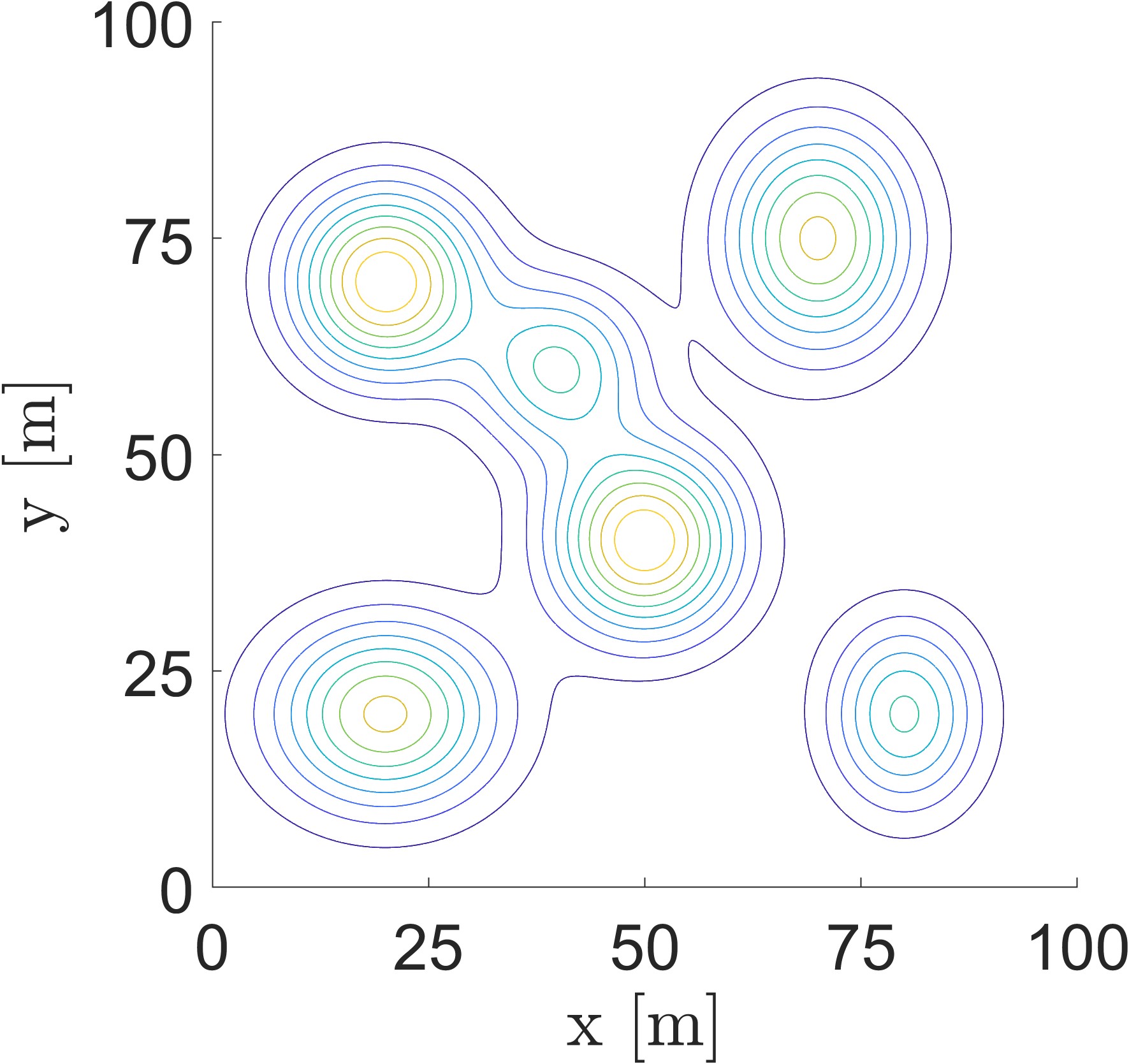}}\quad
    \subfloat[Sample-point representation]{
    \includegraphics[width=0.45\linewidth]{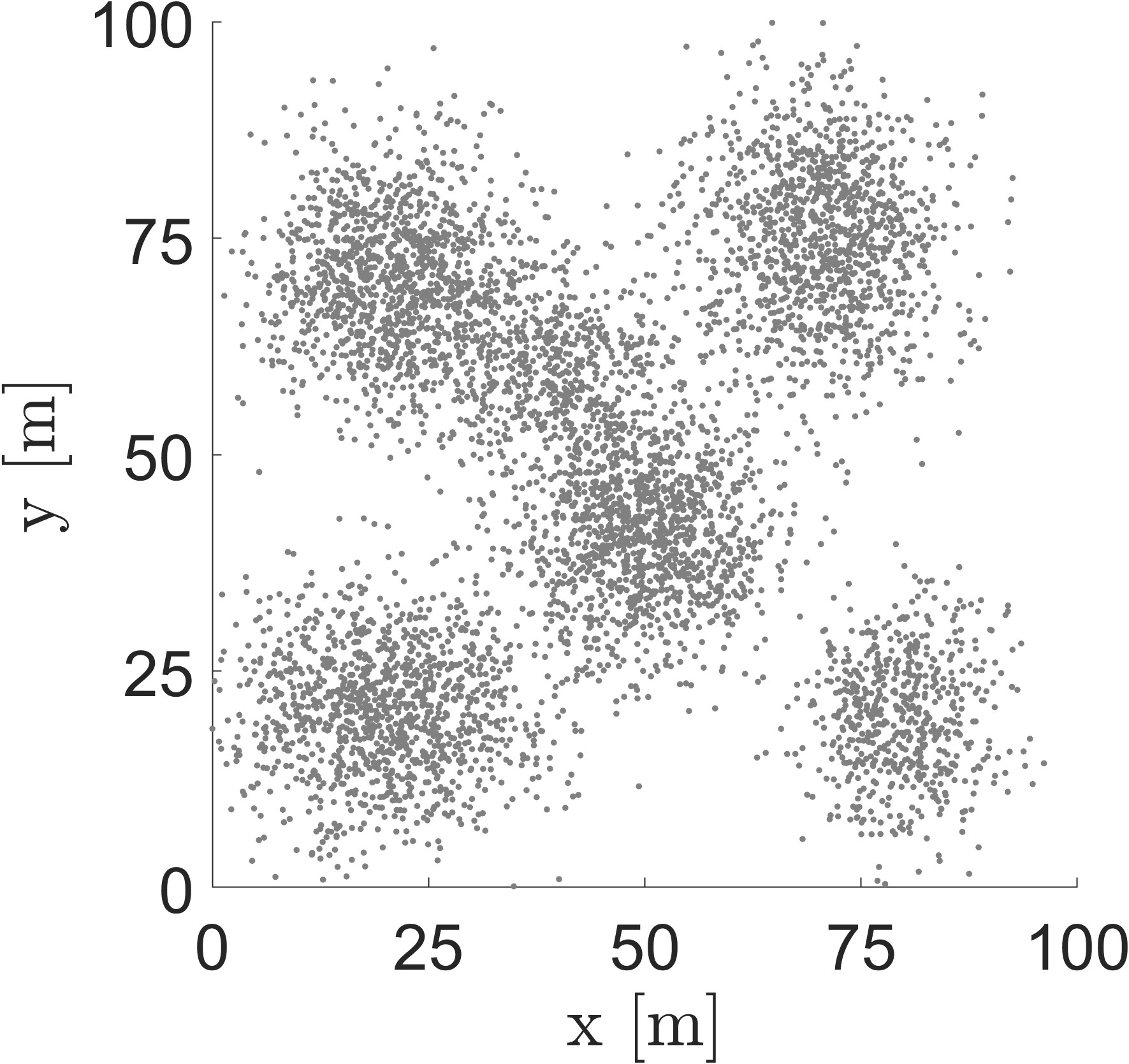}}
    \caption{Reference distribution.}
    \label{fig: refer_dist}
\end{figure}

\subsection{First-order dynamics}
The simulation employed the following first-order dynamics model:
$
    {}^r\xx^{k+1} = {}^r\xx^{k}+{}^ru^{k},$ and $
    {}^r\y^{k} = {}^r\xx^{k},
$
where ${}^r\y^k={}^r\xx^{k} = \begin{bmatrix} {}^rp^k_x&{}^rp^k_y \end{bmatrix}^\top\in \mathbb{R}^2$ represents the Cartesian coordinates in the $xy$-plane at time $k$,  and ${}^ru^k = \begin{bmatrix} {}^ru_x^k&{}^ru_y^k \end{bmatrix}^\top\in \mathbb{R}^2$ denotes the control input. In this dynamics, the matrix product $CB$ equals $I_2$, indicating that the value of $P$ is 1. Table \ref{table: parameter} presents the simulation parameters used for the implementation of the DPC as well as the simulation environment in detail.
Simulations are conducted both with and without control input constraints. In the unconstrained case, the analytical optimal control input \eqref{eqn: opt_cont_input} is used, whereas in the constrained case, the optimal control input is computed using quadratic programming \eqref{eqn: QP_opt_cont_input}. The input constraint is defined as $C_u {}^ru^k \leq D_u$, where $C_u = \begin{bmatrix}
    -I_2&I_2
\end{bmatrix}^\top$ and $D_u = u_\text{max}\begin{bmatrix}
   1 &1&1&1
\end{bmatrix}^\top$. The value of $u_\text{max}$ is provided in Table \ref{table: parameter}. The MATLAB built-in function \texttt{quadprog}, using an interior-point method, is employed to solve \eqref{eqn: QP_opt_cont_input}, with an average computation time of 1.4 ms. In both cases, the communication ranges are set to infinity, i.e., $d_\text{comm.}\rightarrow \infty$.

\begin{table}[!b]
\centering
\caption{Simulation parameters}\label{table: parameter}
\footnotesize
\renewcommand{\arraystretch}{1.25}  
\begin{tabular}{@{}clll@{}}
\hline
\multicolumn{2}{c}{Parameter} & Symbol & Value \\ \hline
\multirow{5}{*}{Agent} & Time step & $\Delta t$ & 0.1 s \\
                     & Agent-points per agent & ${}^rM$ &  \begin{tabular}[c]{@{}l@{}}1500 (1$^{\text{st}}$-order)\\3000 (LTI)\end{tabular} \\
& Number of agents & $L$ & 3 \\
                       & Input constraint & $u_{\max}$ & 5 \\ \hline
\multirow{2}{*}{Env.} & Total sample-points & $N$ & 5975 \\
                      & Domain size & — & 100 m $\times$ 100 m \\ \hline
\end{tabular}
\end{table}

\begin{figure}[h]
    \centering
    \subfloat[Unconstrained case]{
    \includegraphics[width=0.45\linewidth]{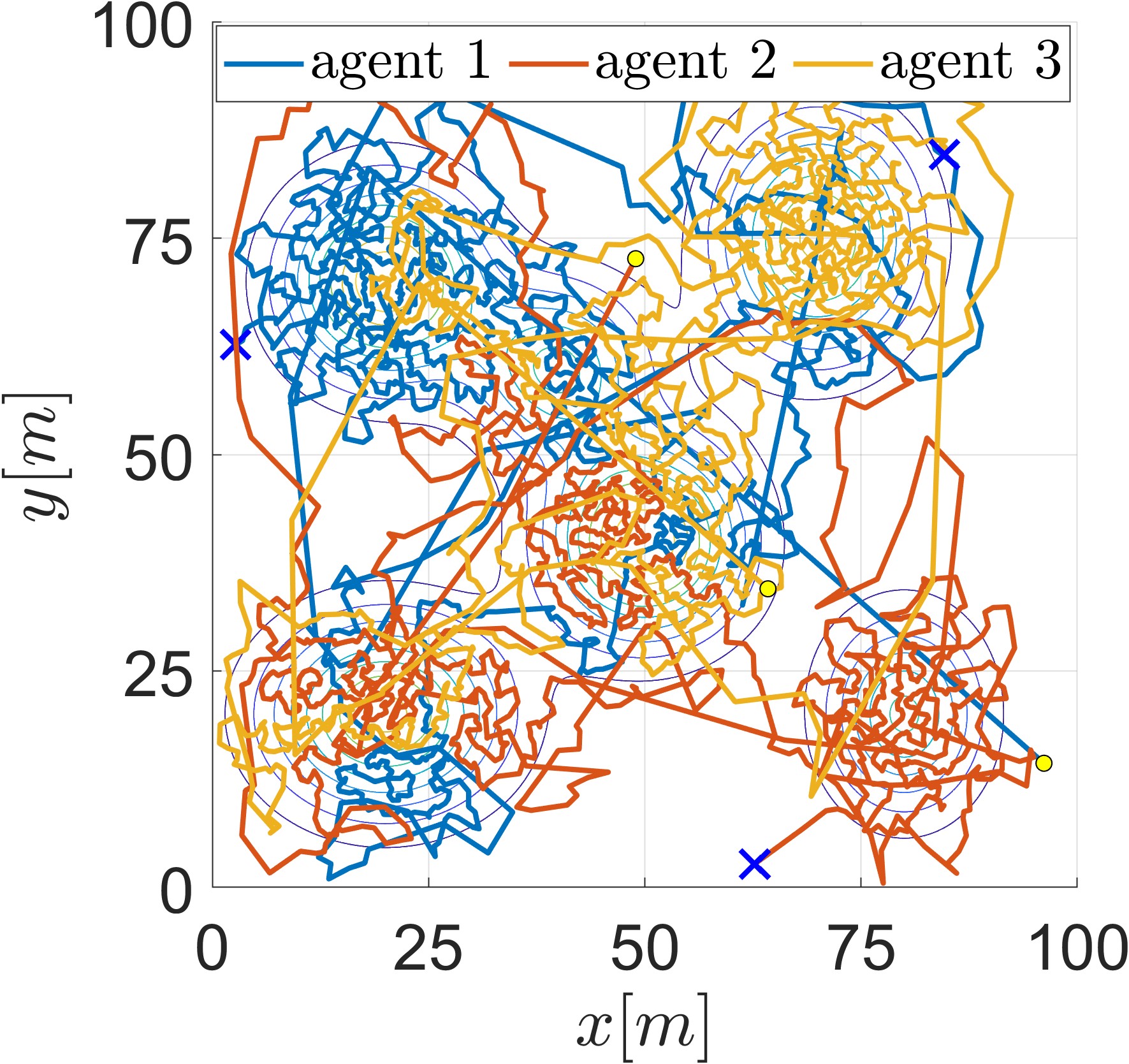}}\quad
    \subfloat[Constrained case]{
    \includegraphics[width=0.45\linewidth]{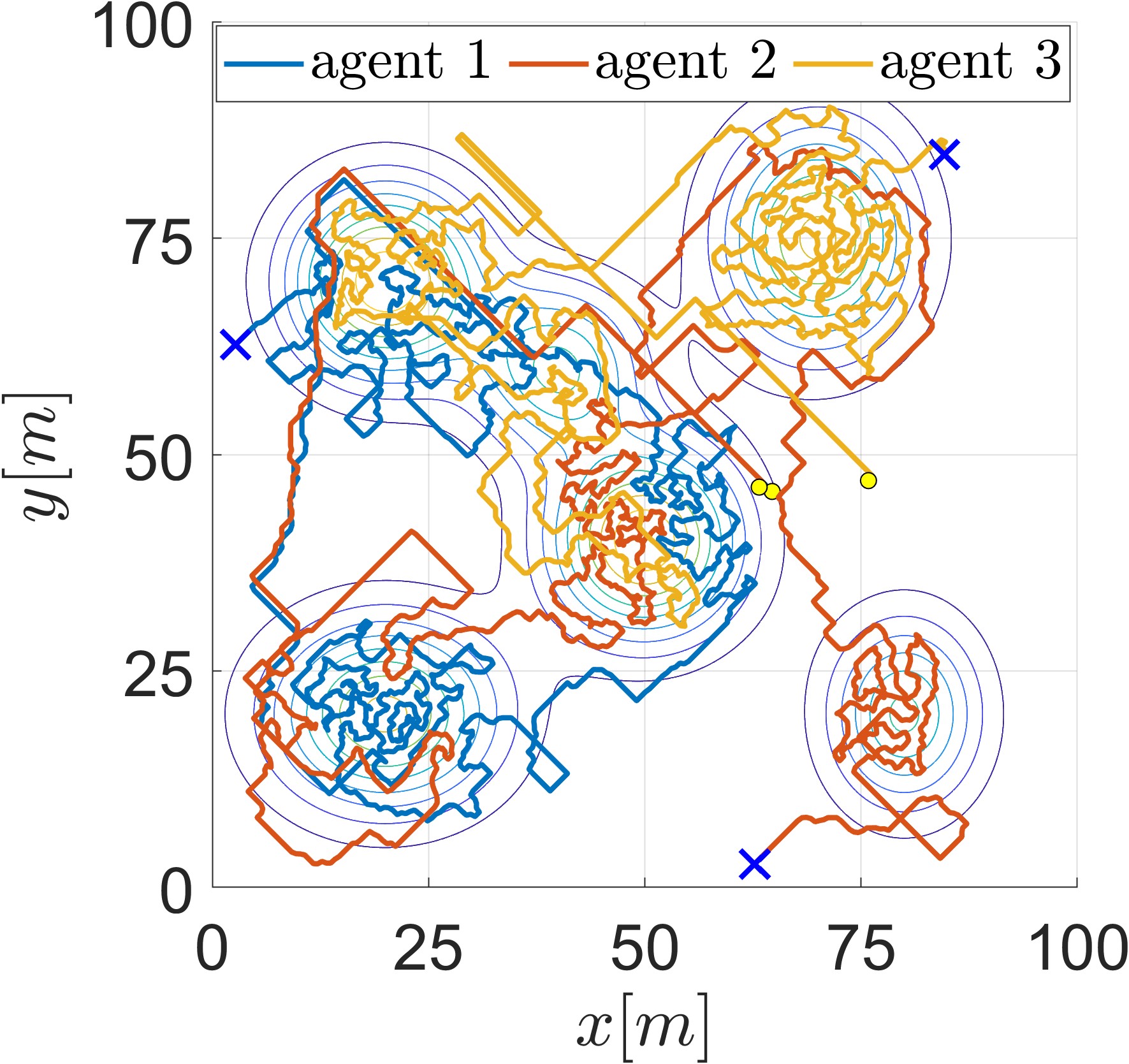}}
    \caption{Agent trajectories generated using the DPC method under first-order dynamics.}
    \label{fig: Traj}
\end{figure}

Figs. \ref{fig: Traj}(a) and (b) depict the agent trajectories for the cases without and with control input constraints, respectively. The blue crosses and yellow circles indicate the agents' initial and terminal locations. Each solid line in red, blue, or yellow represents the trajectory formed by connecting the consecutive agent-points of an agent. In both cases, agent-points are distributed to closely match the sample-point distribution. However, Fig. \ref{fig: Traj}(b) exhibits less thorough exploration than Fig. \ref{fig: Traj}(a) due to the control input constraints limiting the agent's movement.

To validate the convergence behavior of the local Wasserstein distance, $\w^k$ is computed and presented in Fig. \ref{fig: WD_comparison}(a). While convergence is observed over the entire time horizon, variations in $\mathcal{W}^k$ over a selected interval from $k = 821$ to $k = 828$ are highlighted for further illustration. Although the simulation for this figure was conducted using the constrained control input, the optimal control input in \eqref{eqn: opt_cont_input} was also computed at each time step. Using this input, we evaluated how the Wasserstein distance would change if it had been applied. This result is shown as red dashed lines in Fig. \ref{fig: WD_comparison}(a). These inputs were not actually applied during the simulation, so they did not affect the system’s evolution. They are included solely for comparison with the local Wasserstein distance reduction obtained under the constrained control input, which is shown as blue solid lines. 

\begin{figure}
    \centering
    \subfloat[Change in $\w^k$]{
    \includegraphics[width=0.45\linewidth]{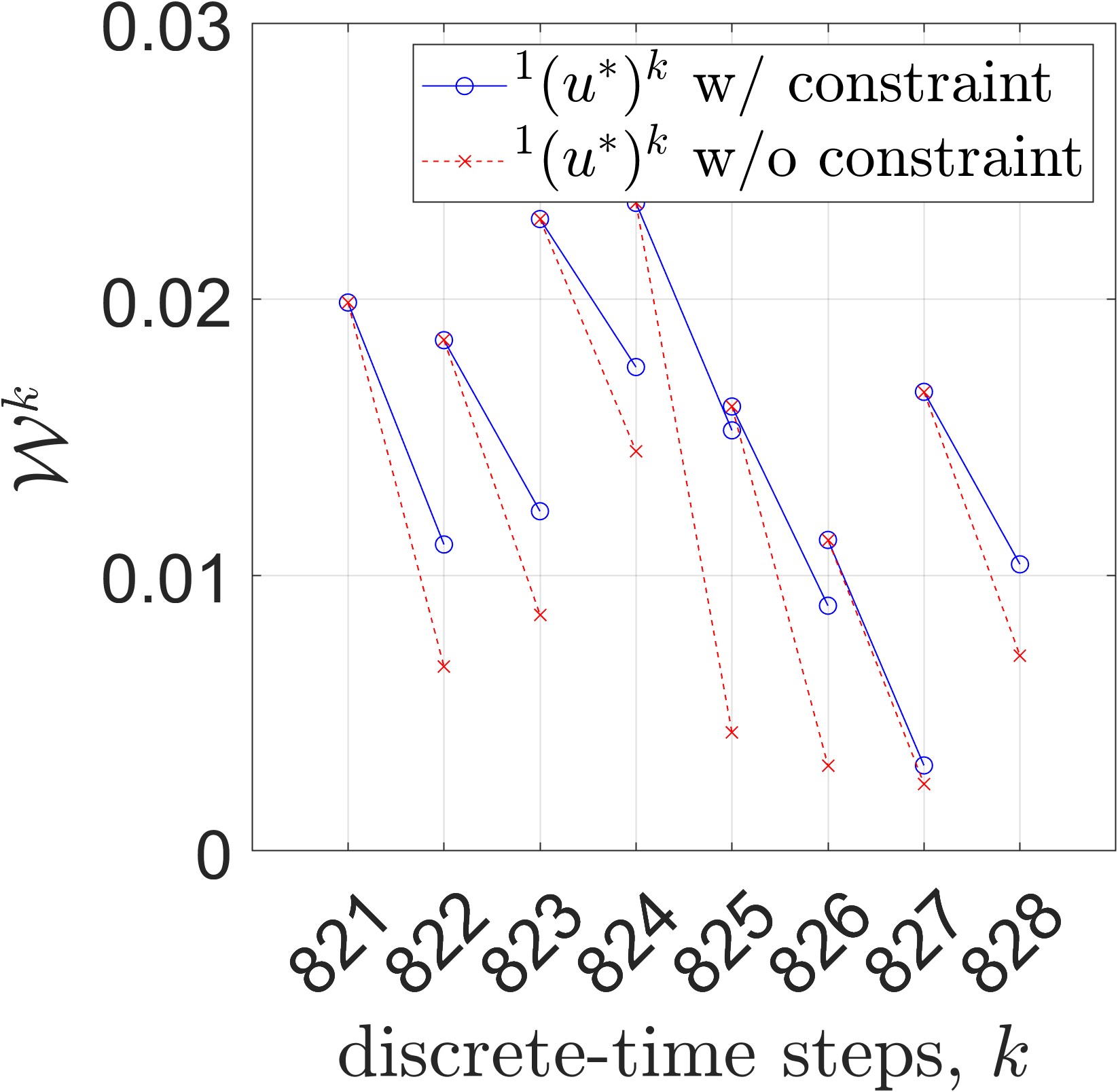}}\quad
    \subfloat[Convergence range of ${}^1u^k$]{
    \includegraphics[width=0.45\linewidth]{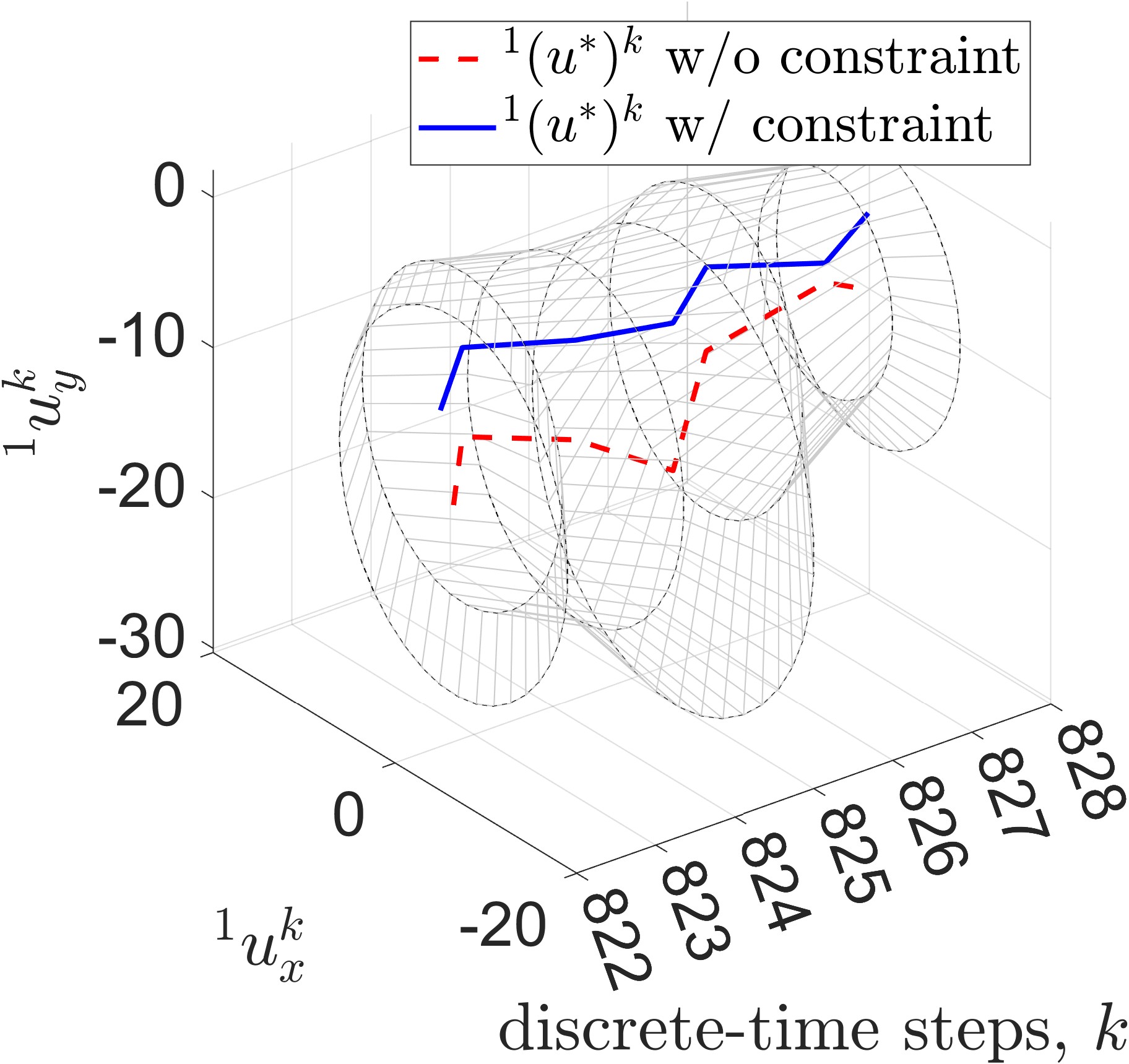}}
    \caption{Change in $\w^k$ and the convergence range of the control input for Agent 1 over discrete-time steps from 821 to 828.}
    \label{fig: WD_comparison}
\end{figure}

Note that Fig. \ref{fig: WD_comparison}(a) depicts the piecewise nature of $\mathcal{W}^k$, which is not continuous between time intervals. This discontinuity arises from the local Wasserstein distance consideration in the proposed DPC. Since the local sample-points, which are used to compute the local Wasserstein distance, change over time steps, $ \mathcal{W}^k$ exhibits a piecewise structure.

Since the feasible set of ${}^ru^k$ under the constraint is a subset of its feasible set without the constraint, the reduction in $\Delta\w^k$ under the constrained case is expected to be less than or equal to that observed in the unconstrained case. Nevertheless, it is noteworthy that the Wasserstein distance for the constrained case still decreases over discrete-time intervals, confirming its convergence.

Fig. \ref{fig: WD_comparison}(b) shows the convergence range of ${}^1u^k$ for Agent 1, obtained from \eqref{eqn: general_conv_range}, over discrete-time steps from $k=821$ to $k=828$. The gray dashed ellipses indicate the convergence range at each time step. The blue solid line represents the optimal control inputs under the input constraint, while the red dashed line represents the unconstrained optimal control inputs, which lie at the centers of the ellipses, where the local Wasserstein distance is minimum. This figure clearly shows that ${}^1({u^*})^k$ with the constraint lies within the convergence range, indicating convergence in the local Wasserstein distance, supporting the main result in Theorem \ref{thm: 1}.
The optimal control input with the constraint may lie outside the convergence range when $\{{}^ru^k\ |\ C_u {}^ru^k \leq D_u,\ {}^ru^k \text{ satisfies equation \eqref{eqn: general_conv_range}}\}=\emptyset$. In this case, the DPC scheme selects the optimal control input that minimizes $\w^k$ within the feasible set of control inputs.

\subsection{Linear Time-Invariant (LTI) system}
The linearized planar drone model \cite{sabatino2015quadrotor} is considered in the simulation to evaluate the performance of the proposed DPC. The state vector is represented by ${}^r\xx^{k} = \begin{bmatrix} {}^r\phi^k&{}^r\theta^k&\odif{^r\phi^k}&\odif{^r\theta^k}&\odif {^rp^k_x}&\odif{^rp^k_y}&{}^rp^k_x&{}^rp^k_y \end{bmatrix}^\top$, the input vector by ${}^ru^k = \begin{bmatrix} {}^r\tau^k_x&{}^r\tau^k_y \end{bmatrix}^\top$, and the output vector by $\begin{bmatrix} {}^rp^k_x&{}^rp^k_y \end{bmatrix}^\top$. Here, $\phi$ and $\theta$ denote the roll and pitch angles, respectively, and $p_x$ and $p_y$ represent the $x$- and $y$-coordinates of the agent, respectively. The symbol $\odif(\cdot)$ in the state vector ${}^r\xx^k$ denotes the variation of the corresponding variable between two consecutive discrete-time steps. The control inputs $\tau_x$ and $\tau_y$ are the torques about the $x$-axis and the $y$-axis, respectively. Since the matrix products $CB=CAB=CA^2B=\mathbf{0}$ while $CA^3B\neq \mathbf{0}$, the value $P$ of this system is 4. To reflect realistic constraints, the bounds are imposed on the state vector and control input of the drone model. Specifically, the state variables were constrained as follows: $-0.52\text{ rad}\leq \phi, \ \theta \leq 0.52\text{ rad},\ -10.47\text{ rad/s}\leq \odif{\phi},\  \odif{\theta} \leq 10.47\text{ rad/s}$, and $-5\text{ m/s}\leq \odif{p_x},\ \odif{p_y}\leq 5\text{ m/s}$. The control inputs were bounded by $-100 \text{\ Nm}\leq \tau_x,\ \tau_y \leq 100 \text{\ Nm}.$ The simulation parameters are summarized in Table \ref{table: parameter}.

\subsubsection{Performance Comparison between DPC and Spectral Multiscale Coverage (SMC) Methods}\label{sec: Sim, SMC vs DPC}
Fig. \ref{fig: LTI_Traj}(a) shows the trajectories of three drones when the DPC scheme is applied under the sample-point distribution shown in Fig. \ref{fig: refer_dist}. The trajectories demonstrate that the distribution of agent-points closely matches the distribution of sample-points, qualitatively indicating that the DPC scheme effectively covers the given reference distribution.

To evaluate the area coverage performance of the proposed DPC method, the Spectral Multiscale Coverage (SMC) method \cite{GM-IM:11} was employed as a baseline for comparison. The SMC method enables non-uniform area coverage by prioritizing regions of greater importance, as specified by a given probability density function (e.g., Fig. \ref{fig: refer_dist}(a)). It approximates both the agents’ trajectories and the reference distribution using a Fourier cosine series, which is used to compute the ergodic metric. A gradient-based algorithm is then employed to minimize this metric, aligning the time-averaged trajectories with the reference distribution.

The simulation was conducted under the same drone model and conditions as those presented in Table \ref{table: parameter} and Fig. \ref{fig: refer_dist}.
While the SMC method supports the first-order and second-order dynamics, it cannot be directly applied to the LTI model. To address this, at each time step, each agent's desired path over the specified horizon is planned using the SMC method with first-order dynamics, and a model predictive control (MPC) is then employed to track this path under the LTI model. The number of basis functions used in the Fourier cosine series is set to 10, and the horizon length for both the SMC-based path planning and the MPC is set to 15 time steps.

\begin{figure}[h]
    \centering
    \subfloat[DPC method]{
    \includegraphics[width=0.45\linewidth]{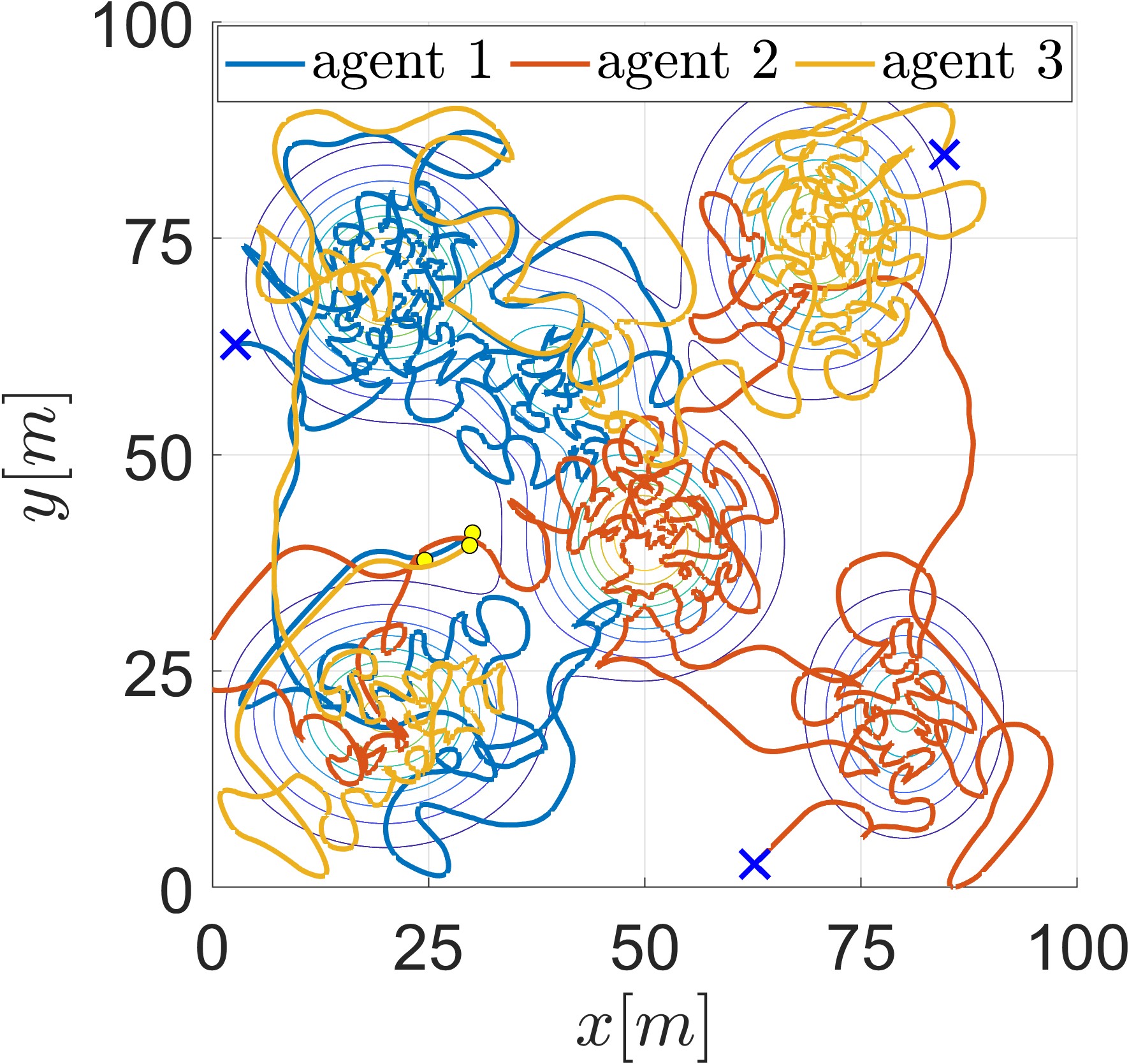}}\quad
    \subfloat[SMC method]{
    \includegraphics[width=0.45\linewidth]{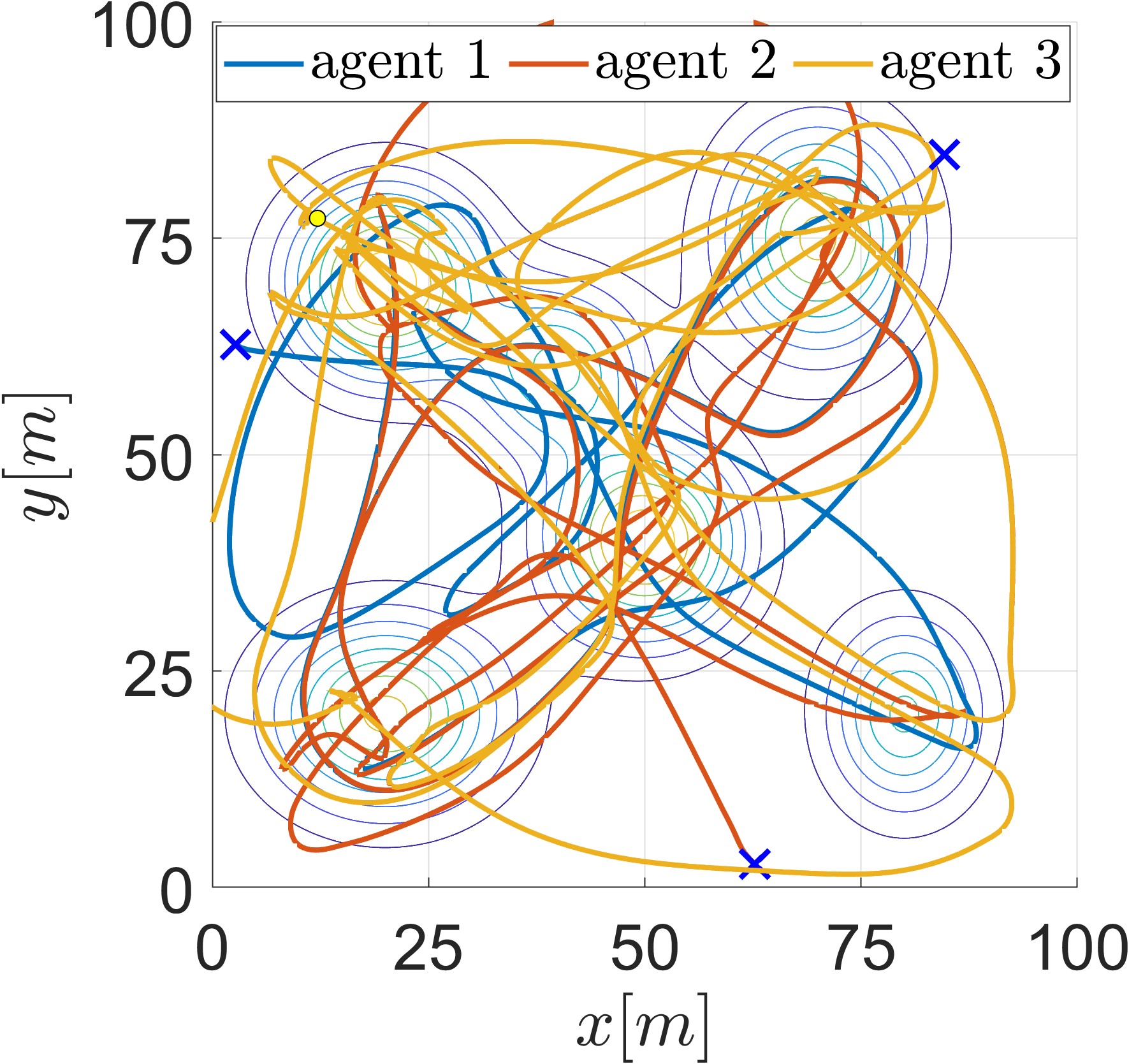}}
    \caption{Agent trajectories generated under the linearized drone system.}
    \label{fig: LTI_Traj}
\end{figure}
\begin{figure}[h]
    \centering
    \subfloat[Change in $\w^k$ of Agent 1 using the DPC method between time steps 822 and 850.]{
    \includegraphics[width=0.45\linewidth]{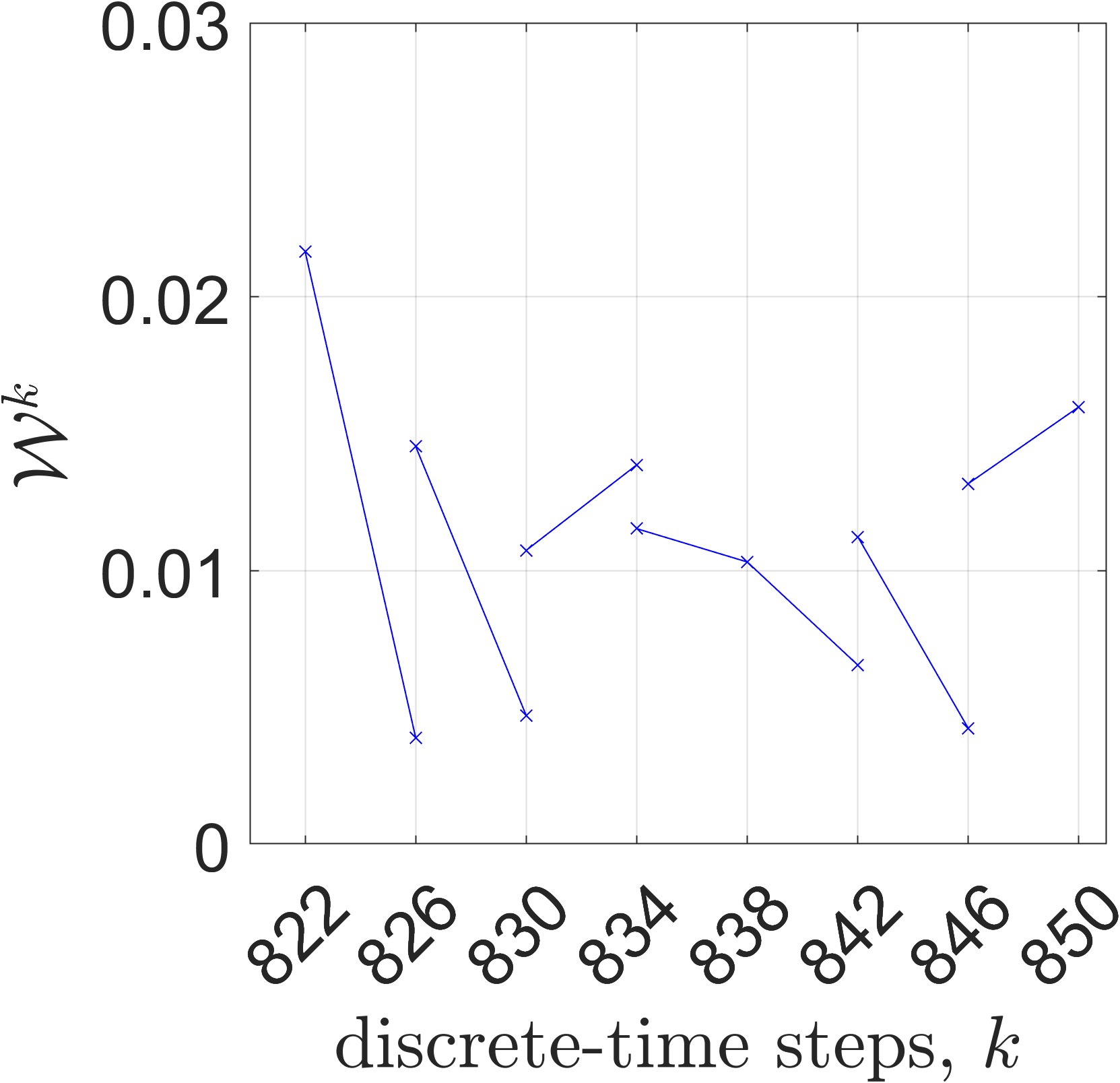}}\quad
    \subfloat[Change in the global Wasserstein distance over the simulation for both the SMC and DPC methods.]{
    \includegraphics[width=0.45\linewidth]{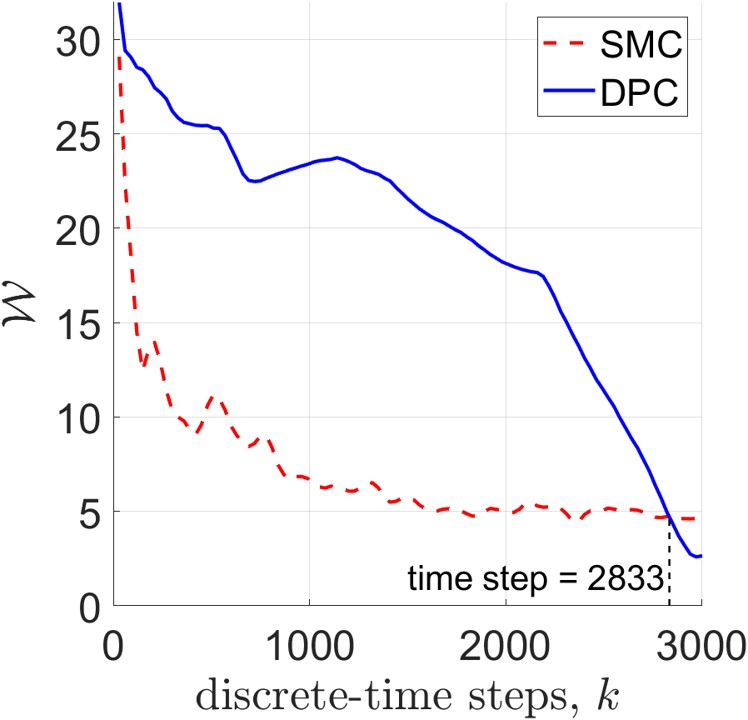}}
    \caption{Change in the local and global Wasserstein distances in the simulation with the linearized multi-drone system.}
    \label{fig: LTI_Wass}
\end{figure}

Fig. \ref{fig: LTI_Traj}(b) illustrates the resulting agent trajectories using the SMC method. The blue crosses and yellow circles indicate the agents' initial and terminal locations. Each solid line in red, blue, or yellow represents the trajectory of an individual agent. It is noteworthy that although the agents are coordinated to align with the reference distribution to some extent, their resulting trajectories appear more dispersed than those of the DPC method in Fig. \ref{fig: LTI_Traj}(a) and exhibit back-and-forth motion between distinct Gaussian distributions. Unlike the DPC framework, which explicitly accounts for each agent’s estimated operation time, the SMC method does not incorporate this constraint. As a result, while the SMC method can eventually lead to improved coverage given sufficient time, there is no formal guarantee on convergence behavior or efficiency, and it may require significantly longer durations to achieve satisfactory performance.

Fig. \ref{fig: LTI_Wass}(a) presents the quantitative measure of $\w^k$ over arbitrary discrete-time intervals from $k=822$ to $k=850$. Since the local Wasserstein distance is designed to converge in $P$ discrete-time steps ($P=4$ for this simulation) under the convergence condition \eqref{eqn: general_conv_range}, $\mathcal{W}^k$ was 
compared to $\mathcal{W}^{k+4}$ in the plot. The figure shows that the local Wasserstein distance generally decreases within each 4-step interval, aligning with the expected convergence behavior. However, occasional increases are observed due to the presence of state constraints, which limit the ability of the control inputs to fully reduce the local distance. These constraints may temporarily hinder the agents from moving directly toward the optimal directions, leading to nonmonotone convergence in some intervals.

Fig. \ref{fig: LTI_Wass}(b) demonstrates the evolution of the Wasserstein distance from a global perspective over the simulation time for both the DPC and SMC methods. As time progresses, the Wasserstein distance exhibits a decreasing trend over time for both methods, indicating that the agents' trajectories are gradually converging toward the reference distribution. Although a theoretical proof for the convergence of the global Wasserstein distance is not provided in the paper, Fig. \ref{fig: LTI_Wass}(b) is included to empirically support its convergence trend in conjunction with the convergence of the local Wasserstein distance. It is worth noting that the global Wasserstein distance for the DPC method exhibits slower convergence at the beginning of the simulation compared to the SMC method. However, the DPC method achieves a lower Wasserstein distance than the SMC method after time step 2833. This trend reflects the prioritization strategy underlying the DPC framework. While the SMC method aims to reduce the ergodicity evaluated over the entire domain at each time step, the DPC method progressively decreases the coverage gap based on the priority of the regions. This causes agents to focus on closer and more critical areas first, improving overall coverage efficiency. As a result of this strategy, one of the Gaussian distributions centered at $x=20, y=20$ is covered last, which leads to a dramatic reduction in the Wasserstein distance after time step 2000. Ultimately, the DPC method achieves better alignment with the reference distribution than the SMC method. At the end of the simulation, the Wasserstein distance for the DPC method is 2.6504, while that for the SMC method is 4.7360, which is approximately 179\% of the DPC value. These results indicate that the DPC method achieves more effective coverage alignment with the reference distribution compared to the SMC method.

\subsubsection{Scalability Performance Evaluation of the DPC Method}
To evaluate the scalability performance of the DPC method, the average stage runtime was measured with a varying number of agents. The same simulation setup as in Section 4.2.1 is used, with the number of agents varied among 1, 2, 4, 8, 16, and 32. For simplicity, we assume the communication is well scheduled with no queueing delay, and the round-trip transmission time is 10 $\pm$ 2 ms. Table \ref{table: comp_time} summarizes the average stage runtime and agent count. As discussed in Remark 6, owing to the fully decentralized framework, Stages A and B do not show a considerable increase in runtime as the number of agents increases. In contrast, Stage C exhibits an approximately linear increase as the number of agents grows, mainly due to communication overhead, which can be mitigated using techniques such as range-limited communication or event-triggered communication strategies.

\begin{table}[t]
\centering
\caption{Average stage runtime (ms) with varying number of agents.}
\label{table: comp_time}

\begin{tabular}{@{}lcccccc@{}}
\hline
 & \multicolumn{6}{c}{\textbf{Number of agents}} \\
\cline{2-7}
 & \textbf{1} & \textbf{2} & \textbf{4} & \textbf{8} & \textbf{16} & \textbf{32} \\
\hline
Stage A & 1.36 & 1.35 & 1.19 & 1.14 & 1.15 & 1.15 \\
Stage B & 0.30 & 0.29 & 0.27 & 0.26 & 0.26 & 0.26 \\
Stage C & 0 & 9.71 & 35.03 & 78.28 & 177.7 & 329.6 \\
\hline
\end{tabular}
\end{table}

\section{Conclusion}
This paper introduces the DPC scheme as a novel approach to the multi-agent non-uniform area coverage problem, leveraging optimal transport theory. Based on a reference distribution, a three-stage control scheme for DPC was developed, consisting of the optimal control, weight update, and weight-sharing stages. The convergence condition of the control input is formulated through the Wasserstein distance, and both analytical and numerical methods for deriving the optimal control law are presented for unconstrained and constrained cases.

Simulations validate the performance of the control input in both unconstrained and constrained scenarios, with the convergence range effectively demonstrating the alignment of the local Wasserstein distance. Additionally, the agent trajectories confirm that DPC successfully drives the agents to cover the reference point distribution, with their paths closely matching the target distribution. The comparison of Wasserstein distances with the SMC method further demonstrates that the DPC method achieves superior coverage performance.
This shows that DPC not only excels in achieving non-uniform area coverage but also offers substantial practical benefits, especially in applications like search and rescue and environmental monitoring, where dynamic, priority-based coverage is essential for operational success.

\begin{acknowledgment}
This work was supported by NSF CAREER Grant CMMI-DCSD-2145810.\end{acknowledgment}

%
% ...

\bibliographystyle{asmems4}

\bibliography{reference.bib}

\appendix       
\vspace{-0.05in}
\section*{Appendix A: Proof of Proposition \ref{prop: simplification of Dwk}}\label{proof: prop 1}
In \eqref{eqn: predictive WD_1}, the predictive Wasserstein distance $(\w^{k+i|k})^2$ is reformulated as
{\small
\allowdisplaybreaks
    \begin{align}
    &(\mathcal{W}^{k+i|k})^2=\sum\nolimits_{q_j\in\s^{k}} \bar{\beta}^k_j \norm{\y^{k+i}- q_j}^2 \nonumber
    \\&\overset{(a)}{=}\alpha^k\norm{\y^{k+i}}^2-2(\y^{k+i})^\top\sum_{q_j\in\s^{k}} \bar{\beta}^k_j  q_j + \sum_{q_j\in\s^{k}} \bar{\beta}^k_j\norm{q_j}^2 \nonumber
    \\&\overset{(b)}{=}\alpha^k\norm{\y^{k+i}}^2 -2\alpha^k(\y^{k+i})^\top \bar{q}^k + \sum\nolimits_{q_j\in\s^{k}} \bar{\beta}^k_j\norm{q_j}^2 \nonumber
    \\&=\alpha^k\norm{\y^{k+i}- \bar{q}^k}^2-\alpha^k\norm{ \bar{q}^k}^2+\sum\nolimits_{q_j\in\s^{k}} \bar{\beta}^k_j\norm{q_j}^2,\label{eqn: simp_wass_1}
\end{align}
}
    where, in equality $(a)$, $\sum\nolimits_{j} \bar{\beta}^{k}_j = \alpha^{k}$ is used, and in equality $(b)$, \eqref{eqn: mass_center_LSP} is used with $k$ instead of $k-1$.

The second and third terms on the right-hand side of \eqref{eqn: simp_wass_1} are simplified by the following equation.
{\small
\begin{align}
    \sum\nolimits_{q_j\in\s^{k}} \bar{\beta}^k_j \norm{q_j -  \bar{q}^k}^2
    = \sum\nolimits_{q_j\in\s^{k}} \bar{\beta}^k_j\norm{q_j}^2-\alpha^k\norm{\bar{q}^k}^2. \label{eqn: simp_centroid}
\end{align}
}
By substituting \eqref{eqn: simp_centroid}, \eqref{eqn: simp_wass_1} is reorganized as
{\small\begin{align}
    (\mathcal{W}^{k+i|k})^2 =\alpha^k\norm{\y^{k+i}- \bar{q}^k}^2+\sum\nolimits_{q_j\in\s^{k}} \bar{\beta}^k_j \norm{q_j -  \bar{q}^k}^2.  \label{eqn: reform_predictive WD}
\end{align}}
    
    Using \eqref{eqn: reform_predictive WD}, \eqref{eqn: diff_Wass} is then rearranged by
    {\small \allowdisplaybreaks
    \begin{align}
    \Delta \w^k&= \sum_{q_j\in \s^{k}} \bar{\beta}^k_j \norm{\y^{k+P}-q_j} ^2 - \sum_{q_j\in\s^{k}} \bar{\beta}^k_j \norm{\y^{k}-q_j}^2\nonumber
    \\&= \alpha^k\norm{\y^{k+P}- \bar{q}^k}^2+\sum\nolimits_{q_j\in\s^{k}} \bar{\beta}^k_j \norm{q_j -  \bar{q}^k}^2\nonumber
    \\&\qquad-\{\alpha^k\norm{\y^{k}- \bar{q}^k}^2+\sum\nolimits_{q_j\in\s^{k}} \bar{\beta}^k_j \norm{q_j -  \bar{q}^k}^2\}\nonumber
    \\&=\alpha^k\left(\norm{\y^{k+P}- \bar{q}^k}^2-\norm{\y^{k}-\bar{q}^k}^2\right). \nonumber
\end{align}\qed}
\vspace{-0.1in}
\section*{Appendix B: Proof of Theorem \ref{thm: 1}}\label{proof: thm 1}
The optimal control input is the control input that minimizes \eqref{eqn: diff_Wass_2.2}, which can be alternatively rewritten with some constant $\D_4\in\mathbb{R}$ by
    {\small\begin{align}
    &\norm{u^k-u^k_o}^2_{\D_1}+\D_4 \label{eqn: diff_Wass_4}
    \\&\quad = (u^k)^\top \D_1 u^k - 2(u_o^k)^\top \D_1 u^k + (u_o^k)^\top \D_1 u_o^k + \D_4 \nonumber
\end{align}}
where $u^k_o\in\mathbb{R}^m$ is a constant vector. 
For the equivalence between \eqref{eqn: diff_Wass_2.2} and \eqref{eqn: diff_Wass_4}, we have
{\small \begin{align}
    &\D_2^\top=-\D_1 u_o^k ,\ \D_4 = -(u_o^k)^\top \D_1 u_o^k +\D_3,\label{eqn: quad_coeff1} 
\end{align}}where $\D_1$ in \eqref{eqn: quad_coeff1} replaces $\D_1^\top$ due to its symmetry. 

\textbf{Case 1 (Unique solution). the rank of $\mathbf{CA^{P-1}B}$ is $\mathbf{p=m}$.} In this case, $\D_1\in\mathbb{R}^{m\times m}$ is invertible, and ${u_o}^k$ in \eqref{eqn: diff_Wass_4} is obtained from \eqref{eqn: quad_coeff1} as follows:
{\small\begin{align}
    u^k_o = -\D_1^{-1}\D_2^\top.\label{eqn: u^k_o}
\end{align} }

Since the matrix $CA^{P-1}B$ is real and has full column rank, the product $(CA^{P-1}B)^\top(CA^{P-1}B)$ is positive definite. Specifically, for all $x \ne \mathbf{0}$, we have
$
x^\top (CA^{P-1}B)^\top (CA^{P-1}B) x = \| CA^{P-1}B x \|^2 > 0,
$
which confirms that $D_1$ is positive definite. Consequently, the optimal control input $(u^*)^k$ that minimizes \eqref{eqn: diff_Wass_4} is uniquely given by $(u^*)^k = u^k_o=-\D_1^{-1}\D_2^\top.$
The convergence range is defined as the region where $\Delta\w^k$ is negative, given by
{\small\begin{align*}
    &\Delta \w^k = \norm{u^k-u^k_o}^2_{\D_1}-(u_o^k)^\top \D_1 u_o^k +\D_3
    \\&\phantom{\Delta \w^k}   =\norm{u^k+\D_1^{-1}\D_2^\top}^2_{\D_1}-\D_2\D_1^{-1}\D_2^\top +\D_3<0
    \\&\rightarrow \norm{u^k+\D_1^{-1}\D_2^\top}^2_{\D_1}<\D_2\D_1^{-1}\D_2^\top -\D_3.
\end{align*}}
\textbf{Case 2 (Non-unique solution). the rank of $\mathbf{CA^{P-1}B}$ is $\mathbf{p<m}$.} In this case, the matrix $\D_1$ in \eqref{eqn: quad_coeff1} is not invertible and hence, $u^k_o$ cannot be computed as in \eqref{eqn: u^k_o}. By substituting $\D_1$ and $\D_2$ from \eqref{eqn: D1 to D3} into \eqref{eqn: quad_coeff1}, the equation is given by $-(CA^{P-1}B)^\top CA^{P-1}Bu^k_o = (CA^{P-1}B)^\top\{\xx^kA^PC - \bar{q}^k\}$. 

Since both sides of this equation lie in the same \( p \)-dimensional subspace \( V \), defined as the subspace generated by the column vectors of the matrix \( (CA^{P-1}B)^\top \), and given that  
$
\text{rank}(CA^{P-1}B) = \text{rank}((CA^{P-1}B)^\top C A^{P-1}B) = \text{rank}(\D_1) = p,
$  
it follows that the column vectors of \( \D_1 \) span \( V \). Consequently, the solution \( u^k_o \) to \eqref{eqn: quad_coeff1} must exist. The solution to \eqref{eqn: quad_coeff1} is given by the Moore-Penrose inverse as
{\small \begin{align}
    u^k_o = -\D_1^+\D_2^\top + (I_m-\D_1^+\D_1)h,\label{eqn: u^k_o_case_2}
\end{align} }
where $h\in\mathbb{R}^m$ is an arbitrary vector, making $u^k_o$ a set-valued solution. 
Knowing the Moore-Penrose inverse properties that $AA^+A = A$ and $AA^+b=b$ for a vector $b$ that  belongs to the column space of $A$, substituting $u^k_o$ into \eqref{eqn: quad_coeff1} for verification leads to $-\D_1 u_o^k = +\D_1\D_1^+\D_2^\top - (\D_1-\D_1\D_1^+\D_1)h = \D_2^\top - (\D_1-\D_1)h = \D_2^\top$.
Therefore, the optimal control input is given by 
    $(u^*)^k =-\D_1^+\D_2^\top + (I_m-\D_1^+\D_1)h$.
Similar to Case 1, the convergence range is given by
$
    \Delta \w^k = \norm{u^k-u^k_o}^2_{\D_1}-(u_o^k)^\top \D_1 u_o^k +\D_3<0,
$
where, by substituting \eqref{eqn: u^k_o_case_2}, the terms $\norm{u^k-u^k_o}^2_{\D_1}$ and $(u^k_o)^\top \D_1 u^k_o$ are rewritten as
{\small\begin{flalign*}
    &\norm{u^k-u^k_o}^2_{\D_1}=
    \norm{u^k+\D_1^+\D_2^\top - (I_m-\D_1^+\D_1)h}^2_{\D_1}
    \\&=\norm{u^k+\D_1^+\D_2^\top}^2_{\D_1}+2(u^k+\D_1^+\D_2^\top)^\top(\D_1-\D_1\D_1^+\D_1)h
    \\&+((I_m-\D_1^+\D_1)h)^\top(\D_1-\D_1\D_1^+\D_1)h=\norm{u^k+\D_1^+\D_2^\top}^2_{\D_1},
    \\&(u^k_o)^\top \D_1 u^k_o =h^\top (I_m - \D_1 \D_1^+) \D_1 (I_m - \D_1 \D_1^+) h&&\nonumber
    \\&- 2h^\top (I_m - \D_1 \D_1^+) \D_1 \D_1^+ \D_2^\top + \D_2 \D_1^+ \D_1 \D_1^+ \D_2^\top
    \\&=\D_2 \D_1^+ \D_2^\top. &&\label{eqn: quad_coeff2_term}
\end{flalign*}}
Finally, the convergence range is simplified as
{\small\begin{flalign*}
    &\Delta \w^k 
    =\norm{u^k+\D_1^+\D_2^\top}^2_{\D_1}-\D_2 \D_1^+ \D_2^\top+\D_3 < 0&&
    \\&\rightarrow \norm{u^k+\D_1^+\D_2^\top}^2_{\D_1} < \D_2 \D_1^+ \D_2^\top-\D_3.&&
\end{flalign*}}
It is noteworthy that the optimal control input and the convergence range for Case 1 are special cases of those for Case 2. Specifically, if the matrix $\D_1$ is invertible, then $\D_1^+$ equals $\D_1^{-1}$. Substituting $\D_1^+$ with $\D_1^{-1}$ in \eqref{eqn: opt_cont_input} and \eqref{eqn: general_conv_range} yields equations that are identical to \eqref{eqn: opt_cont_input_unique} and \eqref{eqn: conv_range}. Consequently, \eqref{eqn: opt_cont_input} and \eqref{eqn: general_conv_range} represent the optimal solutions and the convergence range, respectively, for the general case $p\leq m$. \qed

\end{document}